\documentclass{article}
\usepackage{amsmath,amssymb,amsfonts,mathrsfs}
\usepackage[T1]{fontenc}

\usepackage{graphicx,epstopdf}

\newtheorem{theorem}{Theorem}

\newtheorem{definition}[theorem]{Definition}
\newtheorem{example}[theorem]{Example}

\newtheorem{proposition}[theorem]{Proposition}

\newenvironment{proof}[1][Proof]{\textbf{#1.} }{\ \rule{0.5em}{0.5em}}

\begin{document}
\title{An Introduction to Quantum Filtering}
\date{\today }
\author{John Gough}
\maketitle

\begin{abstract}
The following notes are based on lectures delivered at the research school Modeling and Control of Open Quantum Systems
(Mod\'{e}lisation et contr\^{o}le des syst\`{e}mes quantiques ouverts) at CIRM, Marseille, 16-20 April, 2018, as part of the Trimester \textit{Measurement and Control of Quantum Systems: Theory and Experiments} organized at Institut Henri Poincar\'{e}, Paris, France. The aim is to introduce quantum filtering to an audience with a background in either quantum theory or classical filtering. 
\end{abstract}

\section{Introduction}
Nonlinear filtering theory is a well-developed field of engineering which is used to estimate unknown quantities in the presence of noise. One of the founders of the field was the Soviet mathematician Ruslan Stratonovich who encouraged his student Viacheslav Belavkin to extend the problem to the quantum domain. Classically, estimation works by measuring one or more variables which are dependent on the variables to estimated, and Bayes Theorem plays an essential role in inferring the unknown variables based on what we measure.

However, the proof of Bayes Theorem requires a joint probability distribution for the unknown variables and the measured ones. Once we go to quantum theory, we have to be very careful as incompatible observables do not possess a joint probability distribution - in such cases, applying Bayes Theorem will lead to erroneous results and is the root of many of the paradoxes in the theory.

Our goal is to go through the basic ideas and we derive only the simplest quantum filter. 

\section{Bayes Theorem}
\subsection{Basic Probability Theory}
\subsubsection{Some Intuitive Ideas About Probability}
We begin with an introduction to some basic probabilistic ideas. Imagine a jar full of 100 jelly beans. We select a jelly bean at random (each one has a 1/100 chance to be the one drawn from the jar.)
The beans come in different colours and textures as detailed in Figure \ref{fig:jar}. If we wish to specific both the colour and the texture, then we end up calculating a joint probability. For instance, the probability that the bean selected is both green and rough is 0.1 since we have 10 rough green beans in our jar of 100. Here we are looking for two things to occur jointly. The probability for the bean to be green is 30/100 = 0.3, and to rough is 50/100 =0.5. These are examples of \textit{marginal probabilities}; so-called as they are obtained by summing the appropriate row or column in the table to get answer in the margins. We will spend some time recalling how Bayes Theorem works classically. The Von Neumann measurement model gives a good illustration of when the estimation principle may be applied in the quantum domain, but we make some comments on the role of the Schr\"{o}dinger and Heisenberg picture. We give a discussion of stochastic processes and the classical filtering problem, before going on to the quantum version.

\begin{figure}[h]
\begin{eqnarray*}
&& \qquad \qquad \qquad \mathbf{Colour} \\
&&\mathbf{Texture}\qquad 
\begin{tabular}{l|l|l|l|l}
& Green & Yellow & Blue &  \\ \hline
Rough & \textbf{10} & \textbf{40} & \textbf{0} & 50 \\ \hline
Smooth & \textbf{20} & \textbf{10} & \textbf{20} & 50 \\ \hline
& 30 & 50 & 20 & 100
\end{tabular}
\end{eqnarray*}
\caption{Colours and textures of jelly beans in a jar.}
\label{fig:jar}
\end{figure}

Note that if we only had the marginal probabilities, then we do not have enough information to reconstruct the joint probabilities. In this problem we have Prob$\big\{$Rough$\big\}$ =0.5, Prob$\big\{$Smooth$\big\}$ =0.5, while Prob$\big\{$Green$\big\}$ =0.3, Prob$\big\{$Yellow$\big\}$ =0.5, Prob$\big\{$Blue$\big\}$ =0.2.

Let us suppose that we only knew the marginals. If we were asked to guess what proportion of the beans were both rough and green, say, then we might argue as follows: half the beans are rough; 30 out of 100 are green; so, all things being equal, 15 of the 30 green beans are rough; ergo the proportion of rough green beans is 15/100. But the joint probability is Prob$\big\{$Rough \& Green$\big\}$ =0.20 nor 0.15, so all things are not equal! What this means is of huge importance\footnote{It amounts to a huge hill of beans - sorry, I couldn't resist!}. The \guillemotleft all things being equal\guillemotright \, assumption amounts to what is known in probability theory as \textit{(statistically) independence}. We assume that the variation of one variable\footnote{In the present case we are talking about variation over a descriptive feature, so the variable is a characteristic. In what follows, we will be interested in }, say texture, is uniform over an other, here colour. It is a two way thing! If variable $X$ is independent of variable $Y$, then $Y$ must be independent of $X$ too. It has to be symmetric between $X$ and $Y$.

The fact that colour and texture are not statistically independent means that information about one is useful in working out the chances of the other. Let's suppose someone offers bets on the various colour and that you get to draw the bean from the jar - if no-one sees the colour (including yourself) then the chances for green, yellow and blue are just the marginals. But as you have the bean in your hand, you can tell whether its rough or smooth.
If it's rough then you know that there's no point betting on blue no matter how good the odds are - there are no rough blue jelly beans! You work with the conditional probabilities for the colours given the texture, while everyone else works with just the marginal probabilities.

Note that independence here is not the direct causal dependence that one might be familiar with from physics. It has to do with the distribution. Let suppose that a second jar was filled as follows (Figure \ref{fig:jarSI}).

\begin{figure}[h]
\begin{eqnarray*}
&& \qquad \qquad \qquad \mathbf{Colour} \\
&&\mathbf{Texture}\qquad 
\begin{tabular}{l|l|l|l|l}
& Green & Yellow & Blue &  \\ \hline
Rough & \textbf{10} & \textbf{22} & \textbf{8} & 40 \\ \hline
Smooth & \textbf{15} & \textbf{33} & \textbf{12} & 60 \\ \hline
& 25 & 55 & 20 & 100
\end{tabular}
\end{eqnarray*}
\caption{Colours and textures of jelly beans in a second jar.}
\label{fig:jarSI}
\end{figure}
This time, all the rows are in proportion (5:11:4), and automatically all the columns (2:3). The information that the jelly bean selected from this jar is green, for instance, does not change your probability for it to be rough - it's 10/25 which is the same as you would calculated if you didn't know the colour, 40/60.

\subsubsection{Some Not So Intuitive Ideas}

The axiomatic formulation of probability theory given by Kolmogorov is as follows. One first collects all possible outcomes into a set, $S$, called the \textit{sample space}, the assign probabilities to specified subsets. The allowed subsets are known as \textit{events} and are required to form a $\sigma$-algebra of subsets, $\mathcal{E}$, of the sample space- a standard construct from the branch of mathematics known as measure theory. 

Technically, $\mathcal{E}$ is a $\sigma$-algebra if it is a collection of subsets of $S$ such that $\emptyset \in \mathcal{E}$, if $A \in \mathcal{E}$ then its compliment $A^\prime = \{ \omega \in S: \omega \notin A \}$ is also in $\mathcal{E}$, and if $\{ A_n \}$ is an at most countable number of events in $\mathcal{E}$ then so too is their intersection $\cap_n A_n$ and union $\cup_n A_n$.

Probability is then an assignment of a probability $\mathbb{P} [A] \ge 0$ to each event $A \in \mathcal{E}$ with the rule that $\mathbb{P} [ S] =1$ and $ \mathcal{P} [ \cap_n A_n ] = \sum_n \mathcal{P} [A_n]$ for any at most countable number of events, $\{ A_n \}$ that are non-overlapping (i.e., $A_n \cap A_m = \emptyset$ if $n \neq m$).

Therefore, probability theory is realized as a of special case of measure theory where the measure $\mathbb{P}$ has maximum value $\mathbb{P} [S]=1$. However, there is more too it than that. We also get the definition of conditional probabilities: the probability of event $A$ given that $B$ has occurred is
\begin{eqnarray*}
\mathbb{P} [A|B ] = \frac{  \mathbb{P} [ A \cap B ]}{ \mathbb{P} [B] }
\end{eqnarray*}
which is the joint probability, $ \mathbb{P} [ A \cap B ]$, for both $A$ and $B$ to occur divided by the marginal probability $\mathbb{P} [B]$.
(The reader is encouraged to go back to the jelly bean example to see that this formally definition is precisely the same as the intuitive one we did in our heads.)

\subsubsection{Random Variables}
We will now restrict attention to continuous random variables with well-defined probability densities.
A random variable $X$ has probability distribution function (pdf) $\rho _{X}$
so that 
\begin{eqnarray*}
\Pr \left\{ x\leq X<x+dx\right\} =\rho _{X}\left( x\right) \,dx.
\end{eqnarray*}
Normalization requires $\int_{-\infty }^{\infty }\rho _{X}\left( x\right)
dx=1$. If we have several random variables, then we need to specify their 
\textbf{joint probability}. For instance, if we have a pair $X$ and $Y$ then
their joint pdf will be $\rho _{X,Y}\left( x,y\right) $ with 
\begin{eqnarray*}
\rho _{X}\left( x\right) &=&\int \rho _{X,Y}\left( x,y\right) dy,\quad \mathrm{%
(}x\mathrm{-marginal)} \\
\rho _{Y}\left( y\right) &=&\int \rho _{X,Y}\left( x,y\right) dx,\quad \mathrm{%
(}y\mathrm{-marginal)}
\end{eqnarray*}
and 
\begin{eqnarray*}
1=\int \int \rho _{X,Y}\left( x,y\right) dxdy.
\end{eqnarray*}

We say that $X$ and $Y$ are \textbf{statistically independent} if their
joint probability factors into the marginals 
\begin{eqnarray*}
\rho _{X,Y}\left( x,y\right) =\rho _{X}\left( x\right) \times \rho
_{Y}\left( y\right) ,\quad \text{( independence).}
\end{eqnarray*}

More generally, we can work out the conditional probabilities from a joint
probability. The pdf for $X$ given that $Y=y$ is defined to be 
\begin{eqnarray*}
\rho _{X}\left( x|y\right) \triangleq \frac{\rho _{X,Y}\left( x,y\right) }{%
\rho _{Y}\left( y\right) }.
\end{eqnarray*}
In the special case where $X$ and $Y$ are independent we have 
\begin{eqnarray*}
\rho _{X}\left( x|y\right) =\rho _{X}\left( x\right) .
\end{eqnarray*}
In other words, conditioning on the fact that $Y=y$ makes no change to our
knowledge of $X$.

\subsection{Estimation}
\label{sec:estimation}

Let $X$ be some unknown: in fact, not only do we not know its value, we
don't even know its probability distribution. We wish to get some knowledge
about $X$ however by measuring a related variable $Y$.

Our main modeling assumption is that whenever $X$ is known to take a
particular value of $x$, then the conditional pdf for $Y$ is a known
function: we write this as 
\begin{eqnarray*}
\lambda \left( y|x\right) .
\end{eqnarray*}
For fixed $y$, we refer to $\lambda \left( y|x\right) $ as the \textbf{%
likelihood function} of $x$. Note that for each $x$, $\lambda \left(
y|x\right) $ is a pdf in $y$ and so normalized in $y$ for each $x$ fixed: 
\begin{eqnarray*}
\int \lambda \left( y|x\right) dy=1.
\end{eqnarray*}
Though $\lambda \left( y|x\right) $ is not required to be normalized in $x$
for fixed $y$.

\bigskip

We have $\lambda $ which is the conditional probability for measured
variable $Y$ given that the unknown was $X=x$. But we want to solve the
inverse problem, namely to give the conditional probability for the unknown $%
X$ given the fact that we observe $Y=y$.

The problem however is not well-posed. We do not have enough information in
the problem yet to write down the joint probability To remedy this, we
introduce a pdf for $X$ which is our \textbf{a priori} guess: 
\begin{eqnarray*}
\rho _{\mathrm{prior}}\left( x\right) .
\end{eqnarray*}
We then have the corresponding joint probability for $X$ and $Y$: 
\begin{eqnarray*}
\rho _{\mathrm{prior}}\left( x,y\right) =\lambda \left( y|x\right) \times \rho
_{\mathrm{prior}}\left( x\right) .
\end{eqnarray*}
If we subsequently measure $Y=y$ then we obtain the \textbf{a posteriori}
probability 
\begin{eqnarray*}
\rho _{\mathrm{post}}\left( x|y\right)  &=&\frac{\rho _{X,Y}\left( x,y\right) 
}{\rho _{Y}\left( y\right) } \\
&=&\frac{{\lambda \left( y|x\right) \rho _{\mathrm{prior}}\left( x\right) }}{%
\int \lambda \left( y|x^{\prime }\right) \rho _{\mathrm{prior}}\left(
x^{\prime }\right) dx^{\prime }}.
\end{eqnarray*}

\begin{example}
\label{ex:s+n}
Let $X$ be the position of a particle. We measure 
\begin{eqnarray*}
Y=X+\sigma Z
\end{eqnarray*}
where $Z$ is a standard normal variable, called the ``noise'', independent
of $X$. The likelihood function is 
\begin{eqnarray*}
\lambda \left( y|x\right) =\frac{1}{\sqrt{2\pi }\sigma }e^{-\left(
y-x\right) ^{2}/2\sigma ^{2}},
\end{eqnarray*}
that is, if $X=x$ then $Y$ will be normal with mean $x$ and variance $\sigma
^{2}$. If we choose a prior $\rho _{\mathrm{prior}}$ for $X$ then 
\begin{eqnarray*}
\rho _{\mathrm{post}}\left( x|y\right) =\frac{\rho _{\mathrm{prior}}\left(
x\right) e^{-\left( y-x\right) ^{2}/2\sigma ^{2}}}{\int \rho _{\mathrm{prior}%
}\left( x^{\prime }\right) e^{-\left( y-x^{\prime }\right) ^{2}/2\sigma
^{2}}dx^{\prime }}.
\end{eqnarray*}
In the special case where $X$ is assumed to be Gaussian, say mean $\mu _{0}$
and variance $\sigma _{0}^{2}$, we can give the explicit form of the
posterior as Gaussian with mean $\mu _{1}$ and variance $\sigma _{0}^{2}$
where 
\begin{eqnarray*}
\mu _{1} &=&\frac{\sigma _{1}^{2}}{\sigma _{0}^{2}}\mu _{0}+\frac{\sigma
_{1}^{2}}{\sigma ^{2}}y \\
\frac{1}{\sigma _{1}^{2}} &=&\frac{1}{\sigma _{0}^{2}}+\frac{1}{\sigma ^{2}}.
\end{eqnarray*}
\end{example}

\begin{example}[Parameter Estimation]
Suppose we have a coin with an unknown probability, $x$, for heads. We toss
it three times and obtain the sequence $y=HHT$. The likelihood function is
then 
\begin{eqnarray*}
\lambda \left( HHT|x\right) =x^{2}\left( 1-x\right) ,\quad 0\leq x\leq 1.
\end{eqnarray*}
\begin{figure}[h]
\centering
\includegraphics[width=0.4%
\textwidth]{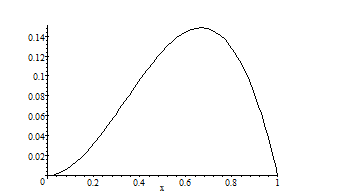}
\caption{The likelihood function of $x$ given the observation $HHT$. The
mode is 2/3.}
\label{fig:likelihood_coin}
\end{figure}

\bigskip

Let us choose the prior to be the uniform distribution $\rho _{X}\left(
x\right) =1$, that is, we take all values for the probability parameter $x$
to be equally likely. A simple calculation gives 
\begin{eqnarray*}
\rho _{\mathrm{post}}\left( x|HHT\right) =\frac{x^{2}\left( x-1\right) }{%
\int_{0}^{1}x^{\prime 2}\left( 1-x^{\prime }\right) dx^{\prime }}%
=12x^{2}\left( 1-x\right) .
\end{eqnarray*}
See Figure \ref{fig:bayes_up_uniform}.
\begin{figure}[tbph]
\centering
\includegraphics[width=0.40%
\textwidth]{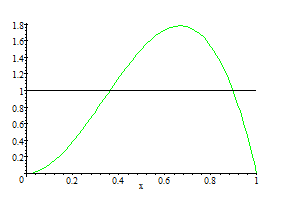}
\caption{The a priori distribution for $x$ (black) and the a posteriori
given $HHT$ (green). A posteriori mean is 3/5 and mode is 2/3.}
\label{fig:bayes_up_uniform}
\end{figure}

If we had however chosen a different prior, we would get a different answer.
For instance, if we set 
\begin{eqnarray*}
\rho _{\mathrm{prior}}\left( x\right) =6x\left( 1-x\right) ,\quad 0\leq x\leq
1,
\end{eqnarray*}
then we calculate 
\begin{eqnarray*}
\rho _{\mathrm{post}}\left( x|HHT\right) =\frac{x^{3}\left( x-1\right) ^{2}}{%
\int_{0}^{1}x^{\prime 3}\left( 1-x^{\prime }\right) ^{2}dx^{\prime }}%
=60x^{3}\left( 1-x\right) ^{2}.
\end{eqnarray*}
This time, see Figure \ref{fig:bayes_up_symm}.

\begin{figure}[tbph]
\centering
\includegraphics[width=0.40%
\textwidth]{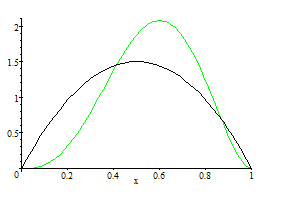}
\caption{The a priori distribution for $x$ (black) and the a posteriori
given $HHT$ (green). A posteriori mean is 4/7 and mode is 3/5.}
\label{fig:bayes_up_symm}
\end{figure}
\end{example}

\section{Quantum Measurement}
\subsection{The Basic Concepts}
The Born interpretation of the wave function, $\psi (x)$, in quantum mechanics is that $|\psi (x)|^2$ gives the probability density of finding the particle at position $x$. More generally, in quantum theory, observables are represented by self-adjoint operators on a Hilbert space. The basic postulate of quantum theory is that the pure states of a system are normalized the wave functions, $\psi$, which we will follow Dirac and denote as kets $|\Psi \rangle$. When we measure an observable, the physical value we record will be an eigenvalue. If the state is $| \Psi \rangle$ then the average value of the observable represented by $\hat A$ is $\langle \hat A \rangle = \langle \Psi |  \hat A | \Psi \rangle$.

Let us recall that a Hermitean operator $\hat{P}$ is called an \textit{orthogonal projection} if it
satisfies $\hat{P}^{2}=\hat{P}$. Then if we have a Hermitean operator $\hat{A}$ with a discrete set of eigenvalues, then there exists
a collection of orthogonal projections $\hat{P}_{a}$ labeled by the
eigenvalues $a$, satisfying $\hat{P}_{a}\hat{P}_{a^{\prime }}=0$ if $a\neq
a^{\prime }$ and $\sum_{a}\hat{P}_{a}=\hat{I}$, such that 
\begin{eqnarray*}
\hat{A}=\sum_{a}a\,\hat{P}_{a}.
\end{eqnarray*}
This is the spectral decomposition of $\hat A$. The operators $\hat{P}_{a}$ project onto $\mathcal{E}_{a}$ which is the 
\textit{eigenspace} of $\hat{A}$ for eigenvalue $a$. In other words, $%
\mathcal{E}_{a}$ is the space of all eigenvectors of $\hat{A}$ having
eigenvalue $a$. The eigenspaces are orthogonal, that is $\langle \psi |\phi
\rangle =0$ whenever $\psi $ and $\phi $ lie in different eigenspaces (this
is equivalent to $\hat{P}_{a}\hat{P}_{a^{\prime }}=0$ if $a\neq a^{\prime }$%
), and every vector $|\psi \rangle $ can be written as a superposition of
vectors $\sum_{a}|\psi _{a}\rangle $ where $|\psi _{a}\rangle $ lies in
eigenspace $\mathcal{E}_{a}$. (In fact, $|\psi _{a}\rangle =\hat{P}_{a}|\psi
\rangle $.)

We note that, for any integer $n$, 
\begin{eqnarray*}
\hat{A}^{n}=\sum_{a}a^{n}\,\hat{P}_{a}
\end{eqnarray*}
and any real $t$
\begin{eqnarray*}
e^{it\hat{A}}=\sum_{a}e^{ita}\hat{P}_{a}.
\end{eqnarray*}

Suppose we prepare a quantum system in a state $|\Psi \rangle $ and perform
a measurement of an observable $\hat{A}$. We know that we may only measure
an eigenvalue $a$ and quantum mechanics predicts the probability $p_{a}$. In
fact, using the spectral decomposition 
\begin{eqnarray*}
\langle \hat{A}^n \rangle =\langle \sum_{a}a^n \,\hat{P}_{a}\rangle
=\sum_{a}\langle a^n\,\hat{P}_{a}\rangle =\sum_{a}a^n \,p_{a},
\end{eqnarray*}
and so 
\begin{eqnarray*}
p_{a}=\langle \hat{P}_{a}\rangle \equiv \langle \Psi |\hat{P}_{a}|\Psi
\rangle .
\end{eqnarray*}

For the special case of a non-degenerate eigenvalue $a$, we have that the
eigenspace $\mathcal{E}_{a}$ is spanned by a single eigenvector $|a\rangle $%
, which we take to be normalized. In this case we have $\hat{P}%
_{a}=|a\rangle \langle a|$

\begin{eqnarray*}
p_{a}=\langle \Psi |\hat{P}_{a}|\Psi \rangle =\langle \Psi |a\rangle \langle
a|\Psi \rangle \equiv \left| \langle a|\Psi \rangle \right| ^{2}.
\end{eqnarray*}
We see that if an observable $\hat{A}$ has a non-degenerate eigenvalue $a$
with normalized eigenvector $|a\rangle $, then if the system is prepared in
state $|\Psi \rangle $, the probability of measuring $a$ in an experiment is 
$\left| \langle a|\Psi \rangle \right| ^{2}$. The modulus squared of an
overlap in this way may therefore have the interpretation as a probability.

The degenerate case needs some more attention. Here the eigenspace $\mathcal{E}_{a}$ can
spanned by a set of orthonormal vectors $|a1\rangle ,|a2\rangle ,\cdots $ so
that $\hat{P}_{a}=\sum_{n}|an\rangle \langle an|$, and so $%
p_{a}=\sum_{n}\left| \langle an|\Psi \rangle \right| ^{2}$. The choice of
the orthonormal basis for $\mathcal{E}_{a}$ is not important!

The probability $p_{a}$ is equal to the length-squared of $\hat{P}_{a}|\Psi
\rangle $, that is, 
\begin{eqnarray*}
p_{a}= \| \hat{P}_{a}\Psi  \| ^{2}.
\end{eqnarray*}
To see this, note that $ \| \hat{P}_{a}\Psi  \| ^{2}$ is the overlap of the
ket $\hat{P}_{a}|\Psi \rangle $ with its own bra $\langle \Psi |\hat{P}%
_{a}^{\dag }$ so 
\begin{eqnarray*}
\| \hat{P}_{a}\Psi \| ^{2}=\langle \Psi |\hat{P}_{a}^{\dag }\,%
\hat{P}_{a}|\Psi \rangle =\langle \Psi |\hat{P}_{a}^{2}|\Psi \rangle
=\langle \Psi |\hat{P}_{a}|\Psi \rangle =p_{a}
\end{eqnarray*}
where we used the fact that $\hat{P}_{a}=\hat{P}_{a}^{\dag }=\hat{P}_{a}^{2}$%
.

In the picture below, we project $|\Psi \rangle $ into the eigenspace $%
\mathcal{E}_{a}$ to get $\hat{P}_{a}|\Psi \rangle $. In the special case
where $|\Psi \rangle $ was already in the eigenspace, it equals its own
projection ($\hat{P}_{a}|\Psi \rangle =|\Psi \rangle $)\ and so $p_{a}=1$
since the state $|\Psi \rangle $ is normalized. If the state $|\Psi \rangle $
is however orthogonal to the eigenspace then its projection is zero ($\hat{P}%
_{a}|\Psi \rangle =0$) and so $p_{a}=0$.

In general, we get something in between. In the picture below we see that $%
|\Psi \rangle $ has a component in the eigenspace and a component orthogonal
to it. The projected vector $\hat{P}_{a}|\Psi \rangle $ will then have
length less than the original $|\Psi \rangle $, and so $p_{a}<1$.

\begin{figure}[tbph]
\centering
\includegraphics[width=1.00\textwidth]{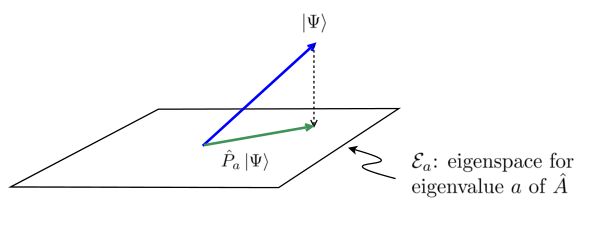} \label%
{fig:Projection_post}
\caption{The state $|\Psi \rangle $ is projected into the eigenspace $%
\mathcal{E}_{a}$ corresponding to the eigenvalue $a$ of $\hat{A}$.}
\end{figure}

\subsubsection{Von Neumann's Projection Postulate}

Suppose the initial state is $|\Psi \rangle $ and we measure the eigenvalue $%
a$ of observable $\hat{A}$ in an given experiment. A second measurement of $%
\hat{A}$ performed straight way ought to yield the same value $a$ again,
this time with certainty.

The only way however to ensure that we measure a given eigenvalue with
certainty is if the state lies in the eigenspace for that eigenvalue. We
therefore require that the state of the system immediately after the result $%
a$ is measured will jump from $|\Psi \rangle $ to something lying in the
eigenspace $\mathcal{E}_{a}$. This leads us directly to the von Neumann
projection postulate.

\bigskip

\textbf{The von Neumann projection postulate:} 
\textit{If the state of a system is given by a ket $|\Psi \rangle $, and a
measurement of observable $\hat{A}$ 
yields the eigenvalue $a$, then the state immediately after measurement
becomes}
$ |\Psi _{a}\rangle =\dfrac{1}{\sqrt{p_{a}}}\,\hat{P}_{a}|\Psi
\rangle .$

We note that the projected vector $\hat{P}_{a}|\Psi \rangle $ has length $%
\sqrt{p_{a}}$ so we need to divide by this to ensure that $|\Psi _{a}\rangle 
$ is properly normalized. The von Neumann postulate is essentially the
simplest geometric way to get the vector $|\Psi \rangle $ into the
eigenspace: project down and then normalize!

\subsubsection{Compatible Measurements}

Suppose we measure a pair of observables $\hat{A}$ and $\hat{B}$ in that
sequence. The $\hat{A}$-measurement leaves the state in the eigenspace of
the measured value $a$, the subsequent $\hat{B}$-measurement then leaves the
state in the eigenspace of the measured value $b$. If we then went back and
remeasured $\hat{A}$ would be find $a$ again with certainty? The state after
the second measurement will be an eigenvector of $\hat{B}$ with eigenvalue $b
$, but this need not necessarily be an eigenvector of $\hat{A}$.

Let $A$ and $\hat{B}$ be a pair of observables with spectral decompositions $%
\sum_{a}a\hat{P}_{a}$ and $\sum_{b}b\hat{Q}_{b}$ respectively. Let us
measure $\hat{A}$ and then $\hat{B}$ recording values $a$ and $b$
respectively. If the initial state was $|\Psi _{\text{in}}\rangle $ then we
obtain after both measurements the final state will be 
\begin{eqnarray*}
|\Psi _{\text{out}}\rangle \propto \hat{Q}_{b}\hat{P}_{a}\,|\Psi _{\text{in}%
}\rangle .
\end{eqnarray*}
In particular $|\Psi _{\text{out}}\rangle $ is an eigenstate of $\hat{B}$
with eigenvalue $b$. However suppose we also wanted $|\Psi _{\text{out}%
}\rangle $ to be an eigenstate of $\hat{A}$ with the original eigenvalue $a$%
, the we must have $\hat{P}_{a}|\Psi _{\text{out}}\rangle =|\Psi _{\text{out}%
}\rangle $ or equivalently 
\begin{eqnarray*}
\hat{P}_{a}\hat{Q}_{b}\hat{P}_{a}\,|\Psi _{\text{in}}\rangle =\hat{Q}_{b}%
\hat{P}_{a}\,|\Psi _{\text{in}}\rangle .
\end{eqnarray*}
If we want this to be true irrespective of the actual initial state $|\Psi _{%
\text{in}}\rangle $ then we arrive at the operator equation 
\begin{eqnarray*}
\hat{P}_{a}\hat{Q}_{b}\hat{P}_{a}=\hat{Q}_{b}\hat{P}_{a}.
\end{eqnarray*}

\begin{proposition}
Let $\hat{P}$ and $\hat{Q}$ be a pair of orthogonal projections satisfying $%
\hat{P}\hat{Q}\hat{P}=\hat{Q}\hat{P}$ then $\hat{P}\hat{Q}=\hat{Q}\hat{P}$.
\end{proposition}

\begin{proof}
We first observe that $\hat{R}=\hat{Q}\hat{P}\hat{Q}$ will again be an
orthogonal projection. To this end we must show that $R^{\dag }=R$ and $%
R^{2}=R$. However, $R^{\dag }=\left( \hat{Q}\hat{P}\hat{Q}\right) ^{\dag }=%
\hat{Q}^{\dag }\hat{P}^{\dag }\hat{Q}^{\dag }=\hat{Q}\hat{P}\hat{Q}=R$ and 
\begin{eqnarray*}
\hat{R}^{2} &=&\left( \hat{Q}\hat{P}\hat{Q}\right) \left( \hat{Q}\hat{P}\hat{%
Q}\right) =\hat{Q}\hat{P}\hat{Q}^{2}\hat{P}\hat{Q} \\
&=&\hat{Q}\hat{P}\hat{Q}\hat{P}\hat{Q}=\hat{Q}(\hat{P}\hat{Q}\hat{P})\hat{Q}
\\
&=&\hat{Q}(\hat{Q}\hat{P})\hat{Q}=\hat{Q}^{2}\hat{P}\hat{Q} \\
&=&\hat{Q}\hat{P}\hat{Q}=\hat{R}.
\end{eqnarray*}

However we also have $\hat{R}=\hat{Q}\hat{P}$, so the relation $\hat{R}=\hat{%
R}^{\dag }$ implies that $\hat{Q}\hat{P}=\hat{P}^{\dag }\hat{Q}^{\dag }=\hat{%
P}\hat{Q}$.
\end{proof}

We see that our operator identity above means that $\hat{Q}_{a}$ and $\hat{P}%
_{b}$ need to commute! If we wanted the $\hat{B}$-measurement not to disturb
the $\hat{A}$-measurement for any possible outcome $a$ and $b$, then we
require that all the eigen-projections of $\hat{A}$ commute with all the
eigen-projections of $\hat{B}$, and this implies that .

\begin{definition}
A collection of observables are compatible if they commute. We define the
commutator of two operators as 
\begin{eqnarray*}
\left[ \hat{A},\hat{B}\right] =\hat{A}\hat{B}-\hat{B}\hat{A}
\end{eqnarray*}
\end{definition}

So $\hat{A}$ and $\hat{B}$\ are compatible if $\left[ \hat{A},\hat{B}\right]
=0$.

\subsection{Von Neumann's Model of Measurement}
The postulates of quantum mechanics outlined above assume that all measurements are idealized, but one might expect the actual process of extracting information from quantum systems to be more involved. Von Neumann modeled the measurement process as follows.
We wish to get information about an observable, $\hat{X}$, say the position of a quantum system. Rather than measure 
$\hat{X}$ directly, we measure an observable $\hat{Y}$ giving the pointer
position of a second system (called the measurement apparatus).

We will reformulate the von Neumann measurement problem in the language of estimation theory from Section \ref{sec:estimation}.
First we assume
that apparatus is described by a wave-function $\phi $. The initial state of
the system and apparatus is $|\Psi _{0}\rangle =|\Psi _{\mathrm{prior}}\rangle
\otimes |\phi \rangle $, i.e., 
\begin{eqnarray*}
\langle x,y|\Psi _{0}\rangle =\Psi _{\mathrm{prior}}\left( x\right) \,\phi
\left( y\right) .
\end{eqnarray*}
(Note  that we are already falling in line with the estimation way of thinking by referring to the initial wave function of the particle as an \guillemotleft \textit{a priori} wave function\guillemotright \, - it is something we have to fix at the outset, even if we recognize it as only a guess for the correct physical state.))
The system and apparatus are taken to interact by means of the unitary 
\begin{eqnarray*}
\hat{U}=e^{i\mu \hat{X}\otimes \hat{P}_{\mathrm{app}}/\hbar }
\end{eqnarray*}
where $\hat{P}_{\mathrm{app}}=-i\hbar \frac{\partial }{\partial y}$ is the
momentum operator of the pointer conjugate to $\hat{Y}$. After coupling, the
joint state is 
\begin{eqnarray*}
\langle x,y|\hat{U}\Psi _{0}\rangle =\Psi _{\mathrm{prior}}\left( x\right)
\,\phi \left( y-\mu x\right) .
\end{eqnarray*}
If the measured value of $\hat{Y}$ is $y$, then the \textbf{a posteriori
wave-function} must be
\begin{eqnarray*}
\psi _{\mathrm{post}}(x|y)=\frac{1}{\sqrt{\rho _{Y}(y)}}\psi _{\mathrm{prior}%
}\left( x\right) \,\phi \left( y-\mu x\right) 
\end{eqnarray*}
where 
\begin{eqnarray*}
\rho _{Y}(y)=\int |\psi _{\mathrm{prior}}\left( x\right) \,\phi \left( y-\mu
x\right) |^{2}dx.
\end{eqnarray*}
Basically, the pointer position will be a random variable with pdf given by $%
\rho _{Y}$: the \textit{a posteriori} wave-function may then be thought of as a random
wave-function on the system Hilbert space:
\begin{eqnarray*}
\psi_{\mathrm{prior}} (x) \longrightarrow \psi _{\mathrm{post}}(x|Y).
\end{eqnarray*}
In the parlance of quantum theorists, the wave function of the apparatus collapses to $| y \rangle$, while we update the \textit{a priori} wave function to get the \textit{a posteriori} one.

We have been describing events in the Schr\"{o}dinger picture where states evolve while observables remain fixed. In this picture, we measure the observable $\hat Y^{\mathrm{in}} =I \otimes \hat Y$.  It is instructive to describe events in the Heisenberg picture. Here the state is fixed as $|\Psi _{0}\rangle =|\Psi _{\mathrm{prior}}\rangle
\otimes |\phi \rangle $, while the observables evolve. In fact, the observable that we actually measure is 
\begin{eqnarray*}
\hat Y^{\text{out}} = \hat U^\ast \big( 
I \otimes \hat Y \big) \hat U = 
 I \otimes \hat Y + \mu \, \hat U^\ast \big( 
 \hat X \otimes I \big)  \hat U ,
\end{eqnarray*}
from which it is clear that we are obtaining some information about $\hat X$.

In fact, the measured observable $ \hat Y^{\text{out}}$ is explicitly of the form \emph{signal}, $\hat U^\ast (\hat X
\otimes I) \hat U$, plus \emph{noise}, $\hat Y^{\text{in}}$ as in Example \ref{ex:s+n}. The noise term, $\hat Y^{\text{in}}$, is independent of the signal and has the prescribed pdf  $ | \phi (y) |^2$.

\section{Stochastic Processes}

\bigskip

\subsection{Noise}

We start with a discrete time model for noise. Suppose we have a sequence $%
\xi _{1},\xi _{2},\xi _{3},\cdots $ of independent random variables
occurring eery $\Delta t$ seconds and with 
\begin{eqnarray*}
\left\langle \xi _{k}\right\rangle =0,\quad \left\langle \xi
_{k}^{2}\right\rangle =1.
\end{eqnarray*}
A random walk, $X_1,X_2,X_3, \cdots $, is given by 
\begin{eqnarray*}
X_{n}=\sum_{k=1}^{n}\xi _{k}
\end{eqnarray*}
and we have 
\begin{eqnarray*}
\left\langle X_{n}\right\rangle =0,\quad \left\langle X_{n}^{2}\right\rangle
=n.
\end{eqnarray*}

\begin{figure}[htbp]
\centering
\includegraphics[width=1.00\textwidth]{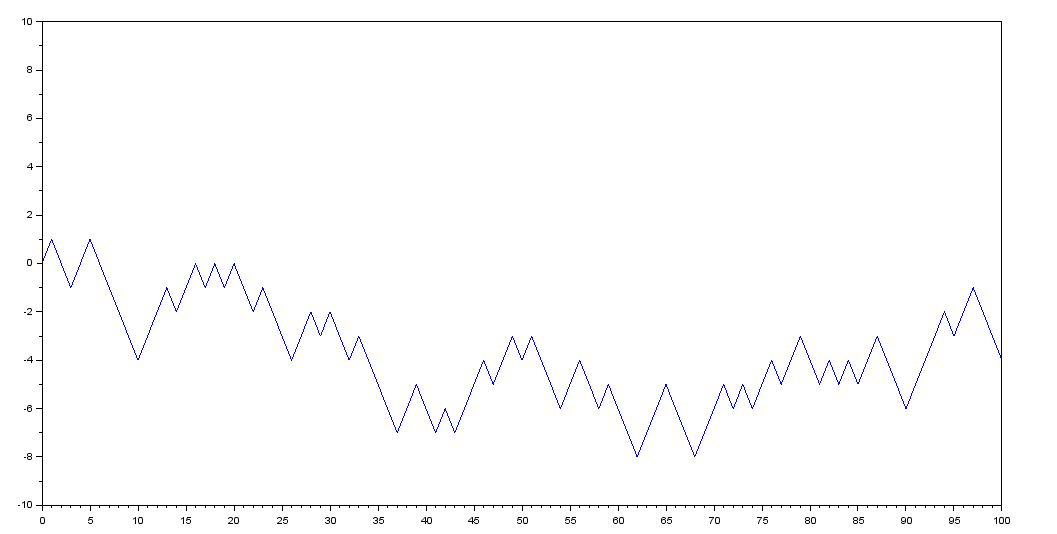}
\caption{The first $n=100$ steps of a random walk where each $\protect\xi$
takes the values $\pm 1$ with equal probability 1/2.}
\label{fig:walk_100}
\end{figure}

For time $t>0$ fixed, let $N\left( t\right) $ be the largest integer less
than or equal to $t/\Delta t$. Introduce the rescaled variable 
\begin{eqnarray*}
W_{\text{approx}}\left( t\right) =\sqrt{\Delta t}\sum_{k=0}^{N\left(
t\right) }\xi _{k}.
\end{eqnarray*}
We have 
\begin{eqnarray*}
\left\langle e^{uW_{\text{approx}}\left( t\right) }\right\rangle
&=&\left\langle e^{u\sqrt{\Delta t}\xi _{1}}\right\rangle \cdots
\left\langle e^{u\sqrt{\Delta t}\xi _{N\left( t\right) }}\right\rangle \\
&=&\left( 1+u\sqrt{\Delta t}\left\langle \xi \right\rangle +\frac{1}{2}%
u^{2}\Delta t\left\langle \xi ^{2}\right\rangle +\cdots \right) ^{N(t)} \\
&\approx &\left( 1+\frac{1}{2}u^{2}\Delta t+\cdots \right) ^{t/\Delta t} \\
&\rightarrow &e^{\frac{1}{2}u^{2}t}\text{ as }\Delta t\rightarrow 0.
\end{eqnarray*}
So $W_{\text{approx}}\left( t\right) $ converges to a limit variable $%
W\left( t\right) $ which is Gaussian with 
\begin{eqnarray*}
\left\langle W\left( t\right) \right\rangle =0,\quad \left\langle W\left(
t\right) ^{2}\right\rangle =t.
\end{eqnarray*}
The family $\left\{ W\left( t\right) :t\geq 0\right\} $ obtained this way is
called a \textbf{Wiener process}.

\begin{figure}[htbp]
\centering
\includegraphics[width=1.00%
\textwidth]{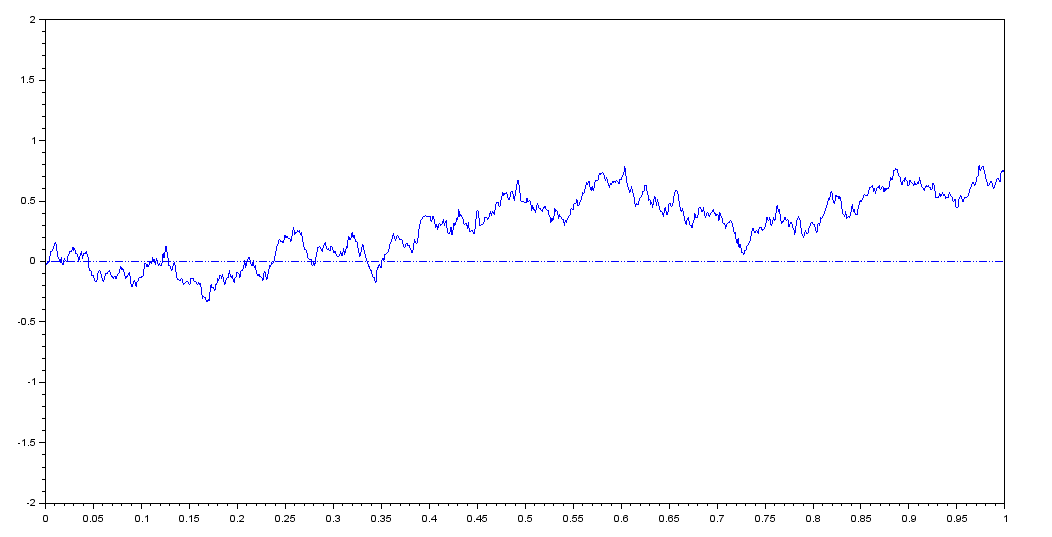}
\caption{A realization of $W_{\mathrm{approx}}$ for $t=1$ and $\Delta t =
1/1000$.}
\label{fig:wiener_1000}
\end{figure}

In Figure \ref{fig:wiener_1000}, we see a typical sample path. We notice
that it looks continuous but rough. In fact, the limit process has sample
paths that are almost always continuous and nowhere differentiable. To see
why, let us look at the approximate derivative 
\begin{eqnarray*}
\dot{W}_{k}=\frac{\Delta W_{k}}{\Delta t} =\frac{W_{k+1} - W_k }{\Delta t} =%
\frac{\sqrt{\Delta t}\, \xi _{k+1}}{\Delta t},
\end{eqnarray*}
then 
\begin{eqnarray*}
\left\langle \dot{W}_{k}\right\rangle =0,\quad \left\langle \dot{W}%
_{k}^{2}\right\rangle =\frac{1}{\Delta t}
\end{eqnarray*}
so the variance of $\dot{W}_{k}$ blows up as $\Delta t\rightarrow 0$.
Formally, one may consider \textbf{white noise} to be the limit process $%
\dot{W}\left( t\right) $ which is Gaussian and $\delta $-correlated: 
\begin{eqnarray*}
\left\langle \dot{W}\left( t\right) \right\rangle =0,\quad \left\langle \dot{%
W}\left( t\right) \dot{W}\left( s\right) \right\rangle =\delta \left(
t-s\right) .
\end{eqnarray*}

\subsection{Random Evolutions}

Let us start with the ODE 
\begin{eqnarray*}
\dot{X}\left( t\right) =v\left( X\left( t\right) \right) ,\quad X\left(
0\right) =x_{0}.
\end{eqnarray*}
To solve this numerically we use a time step $\Delta t$ as before and
consider the discrete time iteration 
\begin{eqnarray*}
X_{k+1}=X_{k}+ v \left( X_{k}\right) \Delta t,\quad X_{0}=x_{0},
\end{eqnarray*}
then $X_{\text{approx}}\left( t\right) =X_{N\left( t\right) }$ should
converge to $X\left( t\right) $ as $\Delta t\rightarrow 0$. In Figure \ref{fig:ODE}, we see a simulation of the ODE $\frac{d}{dt} X =-X$ with initial condition $X(0)=1$. The solution, of course, is just $X(t)= e^{-t}$ and the plot generated is reasonably convincing approximation.

\begin{figure}[htbp]
\centering
\includegraphics[width=1.0\textwidth]{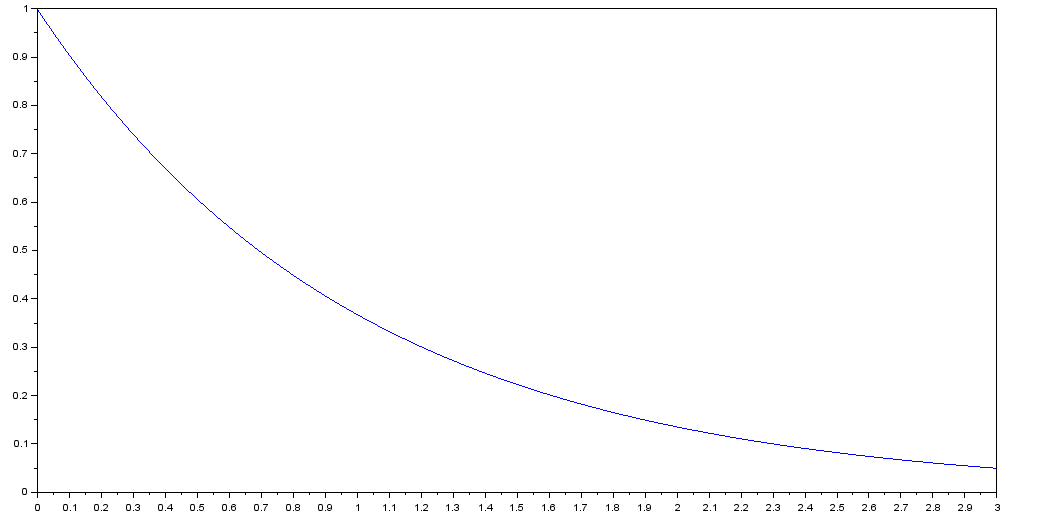}
\caption{An approximation to the solution of the ODE $\dot X = - X$ with $%
X(0)=1$. This should be $e^{-t}$.}
\label{fig:ODE}
\end{figure}

We now try and add some noise: here we consider the Langevin equation
\begin{eqnarray*}
\dot{X}\left( t\right) =v\left( X\left( t\right) \right) +\sigma \dot{W}%
(t),\quad X\left( 0\right) =x_{0}.
\end{eqnarray*}
This time we have the approximation scheme 
\begin{eqnarray*}
X_{k+1}=X_{k}+v\left( X_{k}\right) \Delta t+\sigma \sqrt{\Delta t}\, \xi
_{k+1},\quad X_{0}=x_{0}.
\end{eqnarray*}
A simulation is given below. Here we see a jagged curve replacing our smooth exponential decay.
\begin{figure}[htbp]
\centering
\includegraphics[width=1.00\textwidth]{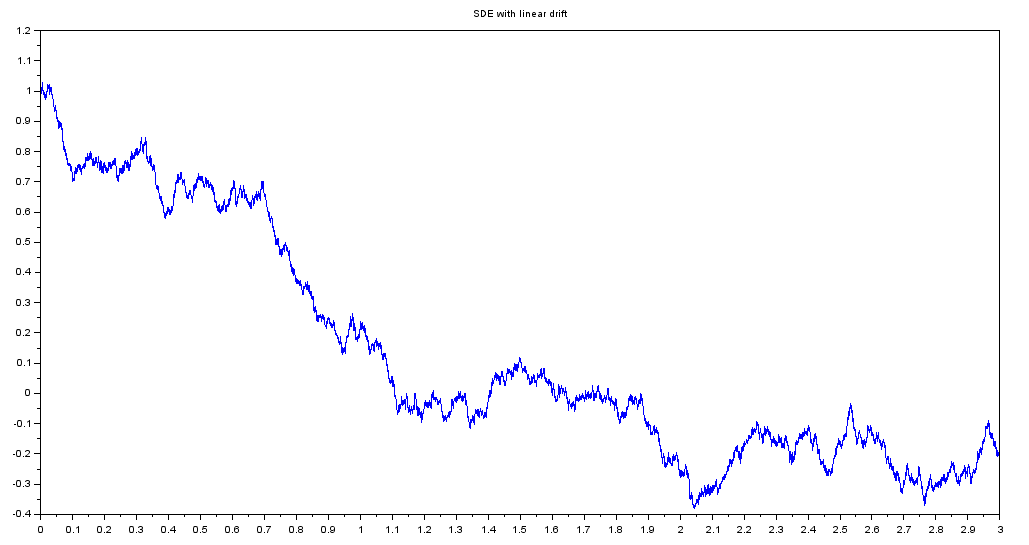}
\caption{Simulation of the SDE $\dot X = -X + \protect\sigma \dot W$ with $%
\protect\sigma = 0.3$. }
\label{fig:SDE}
\end{figure}

So far so good! But if we want to make $\sigma $ depend on $X\left( t\right) 
$ then we need to be more precise. We interpret the SDE 
\begin{eqnarray*}
dX(t)=v(X(t))\,dt+\sigma (X(t))\,dW(t),\quad X(0)=x_{0},
\end{eqnarray*}
to have future pointing differentials, that is 
\begin{eqnarray*}
dX(t)\equiv X(t+dt)-X(t),
\end{eqnarray*}
and is approximated by the scheme 
\begin{eqnarray*}
X_{k+1}=X_{k}+v\left( X_{k}\right) \Delta t+\sigma (X_{k})\sqrt{\Delta t}\, \xi
_{k+1},\quad X_{0}=x_{0}.
\end{eqnarray*}
The limit object, when it exists is referred to as a \textbf{diffusion process}.

A key issue here is that, while 
\begin{eqnarray*}
\left\langle \sigma \left( X_{k}\right) \xi _{k+1}\right\rangle
=\left\langle \sigma \left( X_{k}\right) \right\rangle \left\langle \xi
_{k+1}\right\rangle =0,
\end{eqnarray*}
we have 
\begin{eqnarray*}
\left\langle \sigma \left( X_{k}\right) \xi _{k}\right\rangle \propto \sqrt{%
\Delta t}
\end{eqnarray*}
and so we would get a different limit \ if we used $\xi _{k}$ in the
iteration rather than $\xi _{k+1}$.

\subsection{The Ito Differential}
The differential $dW(t)$ does not behave the way a true infinitesimal should. Its square is not negligible - in fact it is $dt$:
\begin{eqnarray*}
dW(t) \, dW(t) = dt .
\end{eqnarray*}
For instance, when we use the Taylor's Theorem, we will have to go to second order. This is summarized by the Ito formula
\begin{eqnarray*}
d g(W(t)) &=& g^\prime (W(t)) \, dW(t) + \frac{1}{2} 
g^{\prime \prime} (W(t)) \, dW(t)^2 + \cdots \\
&=&
g^\prime (W(t))dW(t) 
+ \frac{1}{2} g^{\prime \prime} (W(t)) dt.
\end{eqnarray*}

To see this in action, consider the Ito integral $\int_0^t W(\tau ) d W(\tau )$ which we may think of as the limit of $\sum_{k=0}^{N(t)} W_k \sqrt{\Delta t} \xi_{k+1}$. The answer is not $\frac{1}{2} W(t)^2$ since $ \langle \int_0^t W(\tau ) d W(\tau ) \rangle =\langle \int_0^t W(\tau )\rangle \langle  d W(\tau ) \rangle=0$ as the future increment is mean zero and independent of the integrand, while $\langle \frac{1}{2} W( t ) ^2 \rangle = \frac{1}{2} t$. We can work out the correct value using Ito's formula for $g(x) = \frac{1}{2} x^2$. Here
\begin{eqnarray*}
d \bigg( \frac{1}{2} W(t)^2 \bigg) =  W(t) dW(t) + \frac{1}{2} dt ,
\end{eqnarray*}
which we can integrate to get
\begin{eqnarray*}
\int_0^t W(\tau ) d W(\tau )= \frac{1}{2} \bigg[ W(t)^ 2 - t \bigg] .
\end{eqnarray*}
Now both sides average to zero!

Returning to the SDE
\begin{eqnarray*}
dX(t)=v(X(t))\,dt+\sigma (X(t))\,dW(t),\quad X(0)=x_{0},
\end{eqnarray*}
we find the equivalent formula
\begin{eqnarray*}
d g(X(t)) &=& g^\prime (X(t)) \, dX(t) + \frac{1}{2} 
g^{\prime \prime} (X(t)) \, dX(t)^2 + \cdots \\
&=&
\big[v(X(t))  g^\prime (X(t)) + 
\frac{1}{2} \sigma (X(t))^2 g^{\prime \prime} (X(t)) \big] dt\\
&&
+ \sigma (X(t))g^{\prime } (X(t))  \, dW(t) .
\end{eqnarray*}
Averaging gives $\big\langle d g(X(t)) \big\rangle
=
\big\langle \mathcal{L} g (X(t)) \big\rangle \, dt$, or
\begin{eqnarray*}
\frac{d}{dt} \big\langle  g(X(t)) \big\rangle
=
\big\langle \mathcal{L} g (X(t)) \big\rangle 
\end{eqnarray*}
where the generator of the diffusion is defined by
\begin{eqnarray*}
\mathcal{L} = v(x) \frac{\partial}{\partial x}  +
 \frac{1}{2} \sigma (x)^2 \frac{\partial^2}{\partial x^2}
.
\end{eqnarray*}
Alternatively, as $\big\langle  g(X(t)) \big\rangle = \int g(x) \rho (x,t) dx$ we may express this as a PDE for $\rho$ known as the Fokker-Planck equation:
\begin{eqnarray*}
\frac{ \partial }{\partial t} \rho (x,t )
=
\mathcal{L}^\star \rho (x,t) =  -\frac{\partial}{\partial x} [v(x) \rho (x) ]
  + \frac{1}{2}  \frac{\partial}{\partial x^2} [\sigma (x)^2 \rho (x) ]
.
\end{eqnarray*}

\begin{example}[Ornstein-Uhlenbeck process]
We consider the SDE
\begin{eqnarray*}
dX = -\gamma X dt + \sigma X dW , \qquad X(0) =x_0 .
\end{eqnarray*}
The noise term is now proportional to $X(t)$ so we need to be careful.

The solution to this equation is
\begin{eqnarray*}
X(t) = x_0  e^{- (\gamma + \frac{1}{2} \sigma^2 )t + \sigma W(t)}
\end{eqnarray*}
which can easily be seen by using the Ito formula. (Exercise)
\end{example}

\subsection{Stochastic Processes}

A \textbf{stochastic process} is a family, $\left\{ X\left( t\right) :t\geq
0\right\} $, of random variables labeled by time. The process is determined
by specifying all the multi-time distributions 
\begin{eqnarray*}
\rho \left( x_{n},t_{n};\cdots ;x_{1},t_{1}\right)
\end{eqnarray*}
for $X\left( t_{1}\right) =x_{1},\cdots ,X\left( t_{n}\right) =x_{n}$ for
each $n\geq 0$.

\bigskip

A stochastic process is said to be \textbf{Markov} if the multi-time
distributions take the form 
\begin{eqnarray*}
\rho \left( x_{n},t_{n};\cdots ;x_{1},t_{1}\right) = T ( x_n , t_n | x_{n-1}
, t_{n-1} ) \cdots T (x_2 , t_2 | x_1 , t_1 ) \, \rho (x_1 , t_1) ,
\end{eqnarray*}
where whenever $t_n > t_{n-1} > \cdots > t_1$.

Here $T(x,t|x_{0},t_{0})$ is the probability density for $X(t)=x$ given that 
$X(t_{0})=x_{0}$, ($t>t_{0}$). 
\begin{eqnarray*}
\text{Prob} \big\{
x \le X(t) \le x + dx | X(t_0) =x_0 \big\}
= T (x , t | x_0 , t_0 ) \, dx ,
\end{eqnarray*}
for $t > t_0$.
It is called the \textbf{transition mechanism}
of the Markov process.

For consistence we should have the following propagation rule, known as the Chapman-Kolmogorov equation in probability theory,
\begin{eqnarray*}
\int T( x, t | x_1 , t_1 )
 \, T (x_1 , t_1 | x_0 , t_0 ) \, dx_1
= T(x,t | x_0 , t_0 ) ,
\end{eqnarray*}
for all $t > t_1 > t_0$.

\bigskip
\begin{example}
The Wiener process (Brownian motion) is determined by 
\begin{eqnarray*}
T\left( x,t|x_{0},t_{0}\right) &=&\frac{1}{\sqrt{2\pi \left( t-t_{0}\right) }%
}e^{-\frac{\left( x-x_{0}\right) ^{2}}{2\left( t-t_{0}\right) }}, \\
\rho \left( x,0\right) &=&\delta _{0}\left( x\right) .
\end{eqnarray*}
The transition mechanism here is the Green's function for the heat equation
\begin{eqnarray*}
\frac{\partial}{\partial t} \rho = \frac{1}{2}
\frac{\partial^2}{\partial x^2} \rho .
\end{eqnarray*}
(In other words, given the data $\rho ( \cdot, t_0) = f(\cdot )$ at time $t_0$, the solution for later times is $\rho (x,t) = \int T(x,t | x_0 , t_0 ) f (x_0) \, dx_0$.)

Norbert Wiener gave an explicit construction - known as the canonical version of Brownian motion, where the sample space is the space of continuous paths, $\mathbf{w}=\left\{ w\left(
t\right) :t\geq 0\right\} $, starting a the origin as sample space, with a suitable $\sigma$-algebra of subsets and
a well defined measure $\mathbb{P}_{\text{Wiener}}^{t}$.
\end{example}

\subsection{Path Integral Formulation}

Indeed, we have 
\begin{eqnarray*}
\rho \left( x_{n},t_{n};\cdots ;x_{1},t_{1}\right) \,dx_{n}\cdots
dx_{1}\propto e^{-\sum_{k}\frac{\left( x_{k}-x_{k-1}\right) ^{2}}{2\left(
t_{k}-t_{k-1}\right) }}dx_{n}\cdots dx_{1}.
\end{eqnarray*}
Formally, we may introduce a limit ``path integral'' with probability
measure on the space of paths 
\begin{eqnarray*}
\mathbb{P}_{\text{Wiener}}^{t}\left[ d\mathbf{w}\right] =e^{-S_{\text{Wiener}%
}\left[ \mathbf{w}\right] }\mathcal{D}\mathbf{w}.
\end{eqnarray*}
where we have the action 
\begin{eqnarray*}
S_{\text{Wiener}}\left[ \mathbf{w}\right] =\int_{0}^{t}\frac{1}{2}\dot{w}%
\left( \tau \right) ^{2}d\tau .
\end{eqnarray*}

For a diffusion $X\left( t\right) $ satisfying 
\begin{eqnarray*}
dX=v\left( X\right) dt+\sigma \left( X\right) dW
\end{eqnarray*}
we have the corresponding measure 
\begin{eqnarray*}
\mathbb{P}_{X}^{t}\left[ d\mathbf{x}\right] =e^{-S_{X}\left[ \mathbf{x}%
\right] }\mathcal{D}\mathbf{x}.
\end{eqnarray*}
where we have the action (substitute $\dot{w}=\frac{\dot{x}-w}{\sigma }$
into $S_{\text{Wiener}}\left[ \mathbf{w}\right] $, and allow for a Jacobian
correction) 
\begin{eqnarray*}
S_{X}\left[ \mathbf{x}\right] =\int_{0}^{t}\frac{1}{2}\frac{[\dot{x}%
-v(x)]^{2}}{\sigma (x)^{2}}d\tau +\frac{1}{2}\int_{0}^{t}\nabla .v(x)d\tau .
\end{eqnarray*}

\section{The Classical Filtering Problem}

Suppose that we have a system described by a process $\left\{ X\left(
t\right) :t\geq 0\right\} $. We obtain information by observing a related
process $\left\{ Y\left( t\right) :t\geq 0\right\} $. 
\begin{eqnarray*}
dX &=& v \left( X\right) dt+\sigma \left( X\right) dW\quad \text{(stochastic
dynamics),} \\
dY &=&h\left( X\right) dt+dZ\quad \text{(Noisy observations).}
\end{eqnarray*}
Here we assume that the dynamical noise $W$ and the observational noise $Z$
are independent Wiener processes.

\subsection{Bayesian Approach}
The joint probability of both $X$ and $Y$ up to time $t$ is 
\begin{eqnarray*}
\mathbb{P}_{X,Y}^{t}\left[ d\mathbf{x},d\mathbf{y}\right] =e^{-S_{X,Y}\left[
x,y\right] }\mathcal{D}\mathbf{x}\mathcal{D}\mathbf{y},
\end{eqnarray*}
where 
\begin{eqnarray*}
S_{X,Y}\left[ \mathbf{x},\mathbf{y}\right]  &=&S_{X}\left[ \mathbf{x}\right]
+\int_{0}^{t}\frac{1}{2}\left[ \dot{y}-h\left( x\right) \right] ^{2}d\tau  \\
&=&S_{X}\left[ \mathbf{x}\right] +S_{\text{Wiener}}[\mathbf{y}]-\int_{0}^{t}%
\left[ h\left( x\right) \dot{y}-\frac{1}{2}h\left( x\right) ^{2}\right]
d\tau ,
\end{eqnarray*}
or 
\begin{eqnarray*}
\mathbb{P}_{X,Y}^{t}\left[ d\mathbf{x},d\mathbf{y}\right] =\mathbb{P}_{X}^{t}%
\left[ d\mathbf{x}\right] \mathbb{P}_{\mathrm{Wiener}}^{t}\left[ d\mathbf{y}%
\right] \,  \lambda \left( \mathbf{y}%
| \mathbf{x} \right) .
\end{eqnarray*}
where the \textbf{Kallianpur-Streibel likelihood}\footnote{Readers with a background in stochastic processes will recognize this as a Radon-Nikodym derivative associated with a Girsanov transformation.} is
\begin{eqnarray*}
\lambda \left( \mathbf{y}%
|\mathbf{x} \right)  =e^{\int_{0}^{t}\left[ h\left( x\right) dy(\tau )-\frac{1}{2}h\left(
x\right) ^{2}d\tau \right] }.
\end{eqnarray*}
The distribution for $X\left( t\right) $ given observations $\mathbf{y}%
=\left\{ y\left( \tau \right) :0\leq \tau \leq t\right\} $ is then
\begin{eqnarray*}
\rho \left( x_{t}|\mathbf{y}\right) =\frac{\int_{ x(0)=x_0}^{ x(t) =x_t } \lambda \left( \mathbf{y}%
| \mathbf{x}\right) \mathbb{P}_{X}^{t}%
\left[ d\mathbf{x}\right] }{\int_{x(0) =x_0} \lambda \left( \mathbf{y}%
| \mathbf{x}^\prime \right) \mathbb{P}_{X}^{t}%
\left[ d\mathbf{x}^\prime \right]}
\end{eqnarray*}

\subsection{The Filter Equations}
Let us write $\rho_t (x) $ for $\rho_t^{\text{post}} (x | \{ Y(\tau ) : 0 \le \tau \leq t\}  )$. This is the pdf for $X(t)$ conditioned on the past observations $\{ Y(\tau ) : 0 \le \tau \leq t\}$.

The estimate for $f (X(t))$ for any function $f$ is called the \textbf{filter} and we may write this as
\begin{eqnarray}
\pi_t (f) =\int \rho_t (x) f(x) \, dx =
 \frac{ \int \sigma_t (x) f(x) dx }{\int \sigma_t (x^\prime )
dx^\prime }
\label{eq:filter_pi}
\end{eqnarray}
where the non-normalized $\sigma_t (x_t) =\int_{ x(0)=x_0}^{ x(t) =x_t } \lambda \left( \mathbf{y}%
| \mathbf{x}\right) \mathbb{P}_{X}^{t}%
\left[ d\mathbf{x}\right]$ can be shown to satisfy the \textbf{Duncan-Mortensen-Zakai equation}
\begin{eqnarray*}
d \sigma_t (x) = \mathcal{L}^\ast \sigma_t (x) \, dt
+h(x) \sigma_t (x) \, dY(t) .
\end{eqnarray*}

The estimate for $f (X(t))$ will be the \textbf{filter}
\begin{eqnarray*}
d \pi_t (f) =\pi_t (\mathcal{L}  f) \, dt + \big\{ \pi_t (fh) - \pi_t (f) \pi_t (h) \big\} dI(t) ,
\end{eqnarray*}
where the \textbf{innovations process} is defined as
\begin{eqnarray*}
dI(t)  = dY(t) - \pi_t (h ) \, dt . 
\end{eqnarray*}

\section{Quantum Markovian Systems}

\subsection{Quantum Systems with Classical Noise}
We consider a quantum system driven by Wiener noise. For $H$ and $R$ self-adjoint, we set
\begin{eqnarray*}
U(t) = e^{-iH t -iR W(t) } ,
\end{eqnarray*}
which clearly defines a unitary process. From the Ito calculus we can quickly deduce the corresponding Schr\"{o}dinger equation
\begin{eqnarray*}
dU(t) = \big[ -iH - \frac{1}{2} R^2 \big] U(t) \, dt
-iR U(t) \, dW(t) .
\end{eqnarray*}
If we set $j_t (X) = U(t)^\ast XU(t)$, which we may think of as an embedding of the system observable $X$ into a noisy environment, then we similarly obtain
\begin{eqnarray*}
dj_t (X) 
= j_t \big( \mathcal{L} (X) \big) \, dt
-i j_t \big(  [X,R] \big) \, dW(t) .
\end{eqnarray*}
where
\begin{eqnarray*}
\mathcal{L} (X) =
-i[X,H] - \frac{1}{2} \big[ [X,R],R \big] .
\end{eqnarray*}

An alternative is to use Poissonian noise. Here we apply a unitary kick, $S$, at times distributed as a Poisson process with rate $\nu>0$. Let $N(t)$ count the number of kicks up to time $t$, then $\{ N(t) : t \ge 0 \}$ is a stochastic process with independent stationary increments (like the Wiener process) and we have the Ito rules
\begin{eqnarray*}
dN(t) \, dN(t) = dN(t), \qquad \langle dN(t) \rangle
= \nu \, dt .
\end{eqnarray*}
The Schr\"{o}dinger equation is $dU(t) = (S-I) U(t) \, dN(t)$ and for the evolution of observables we now have
\begin{eqnarray*}
dj_t (X) = j_t \big( \mathcal{L} (X)\big) dN(t), \qquad
\mathcal{L}(X) = S^\ast X S -X.
\end{eqnarray*}

\subsection{Lindblad Generators}
A quantum dynamical semigroup is a family of CP maps, 
$\{ \Phi_t: t \geq 0\}$, such that $\Phi_t \circ \Phi_s
=\Phi_{t+s}$ and $\Phi (I) =I$. Under various continuity conditions one can show that the general form of the generator is
\begin{eqnarray*}
\mathcal{L} (X) = \sum_k \frac{1}{2} L_k^\ast [X,L_k]
+ \sum_k \frac{1}{2} [L_k^\ast ,X] L_k - i [X,H ].
\end{eqnarray*} 
These include the examples emerging from classical noise above - in fact, combinations of the Wiener and Poissonian cases give the general classical case. But the class of Lindblad generators is strictly larger that this, meaning that we need quantum noise! This is typically what we consider when modeling quantum optics situation.

\subsection{Quantum Noise Models}

\subsection{Fock Space}
We recall how to model bosonic fields. We wish to describe a typical pure state $| \Psi \rangle$ of the field. If we look at the field we expect to see a certain number, $n$, of particles at locations $x_1 , x_2 , \cdots , x_n$ and to this situation we assign a complex number (the probability amplitude) $\psi_n (x_1 , x_2, \cdots x_n ) $. As the particles are indistinguishable bosons, the amplitude should be completely symmetric under interchange of particle identities. 

\begin{figure}[htbp]
	\centering
		\includegraphics[width=0.50\textwidth]{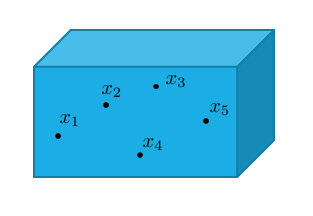}
	\caption{Quantum field in a box - a 5 photon state. }
	\label{fig:field_box}
\end{figure}

The field however can have an indefinite number of particles - that is, it can be written as a superposition of fixed number states. The general form of a pure state for the field will be
\begin{eqnarray*}
| \Psi \rangle = \big( \psi_0 , \psi_1 , \psi_2, \psi_3 ,
\cdots \big).
\end{eqnarray*}
Note that the case $n=0$ is included and is understood as the vacuum state. Here $\psi_0 $ is a complex number, with $ p_0 =| \psi_0 |^2$ giving the probability for finding no particles in the field.

The probability that we have exactly $n$ particles is
\begin{eqnarray*}
p_{n}=\int \left| \psi _{n}\left( x_{1},x_{2},\cdots ,x_{n}\right) \right|
^{2}dx_{1}dx_{2}\cdots dx_{n},
\end{eqnarray*}
and the normalization of the state is therefore $\sum_{n=0}^\infty p_n =1$.

In particular, we take the vacuum state to be
\begin{eqnarray*}
| \Omega \rangle = \big( 1 , 0,0,0, \cdots \big) .
\end{eqnarray*}

The Hilbert space spanned by such indefinite number of 
indistinguishable boson states is called \textbf{Fock Space}.

A convenient spanning set is given by the exponential vectors 
\begin{eqnarray*}
\langle x_{1},x_{2},\cdots ,x_{n}|\exp \left( \alpha \right) \rangle =\frac{1%
}{\sqrt{n!}}\alpha \left( x_{1}\right) \alpha \left( x_{2}\right) \cdots
\alpha \left( x_{n}\right) .
\end{eqnarray*}
They are in fact over-complete and we have the inner products
\begin{eqnarray*}
&&\langle \exp \left( \alpha \right) |\exp \left( \beta \right) \rangle \\
&=&\sum_{n}  \frac{1}{\sqrt{n!}} \int \alpha \left(
x_{1}\right) ^{\ast }\cdots \alpha \left( x_{n}\right) ^{\ast }\beta \left(
x_{1}\right) \cdots \beta \left( x_{n}\right) \,dx_{1}\cdots dx_{n} \\
&=&e^{\int \alpha \left( x\right) ^{\ast }\beta \left( x\right) dx} \\
&=&e^{\langle \alpha |\beta \rangle }.
\end{eqnarray*}
The exponential vectors, when normalized, give the analogues to the coherent states for a single mode.

We note that the vacuum is an example: $| \Omega \rangle = | \exp (0) \rangle$.

\subsection{Quanta on a Wire}
We now take our space to be 1-dimensional - a wire. Let's parametrize the position on the wire by variable $\tau$, and denote by $\mathfrak{F}_{[s,t]}$ the Fock space over a segment of the wire $s \le \tau \le t$. We have the following tensor product decomposition \begin{eqnarray*}
\mathfrak{F}_{A \cup B} = \mathfrak{F}_A \otimes
\mathfrak{F}_B, \qquad \qquad \text{if} A \cap B = \emptyset
.
\end{eqnarray*}

In is convenient to introduce quantum white noises $b(t)$ and $b(t)^\ast$ satisfying the singular commutation relations
\begin{eqnarray*}
 [b(t) ,b(s)^\ast ] &=& \delta (t-s) .
\end{eqnarray*}
Here $b(t)$ annihilates a quantum of the field at location $t$. In keeping with the usual theory of the quantized harmonic oscillator, we take it that $b(t)$ annihilates the vacuum: $b(t) \, | \Omega \rangle  = 0$. More generally, this implies that
\begin{eqnarray}
b(t) \, | \exp (\beta ) \rangle =
\beta(t)\, | \exp (\beta ) \rangle .
\label{eq:eigen_b}
\end{eqnarray}
The adjoint $b(t)^\ast$ creates a quantum at position $t$.

The quantum white noises are operator densities and are singular, but their integrated forms do correspond to well defined operators
which we call the \textbf{annihilation and creation processes}, respectively,
\begin{eqnarray*}
B(t) = \int_0^t b( \tau ) d \tau , \qquad
B(t)^\ast = \int_0^t b(\tau )^\ast d \tau .
\end{eqnarray*}
We see that 
\begin{eqnarray*}
[B(t) , B(s)^\ast ]
= \int_0^t d\tau \int_0^s d \sigma \, \delta (\tau - \sigma )
= \text{min} (t,s).
\end{eqnarray*}

In addition we introduce a further process, called the \textbf{number process}, according to
\begin{eqnarray*}
\Lambda (t) = \int_0^t b(\tau )^\ast b ( \tau ) d \tau .
\end{eqnarray*}

\subsection{Quantum Stochastic Models}
We now think of our system as lying at the origin $ \tau =0$ of a quantum wire. The quanta move along the wire at the speed of light, $c$, and the parameter $\tau$ can be thought of as $x/c$ which is the time for quanta at a distance $x$ away to reach the system. Better still $\tau$ is the time at which this part of the field passes through the system. The process $B(t) = \int_0^t b( \tau ) d\tau$ is the operator describing the annihilation of quanta passing through the system at some stage over the time-interval $[0,t]$.

Fix a system Hilbert space, $\mathfrak{h}_0$, called the \textbf{initial space}. A quantum stochastic process is a family of operators, $ \{ X(t): t \ge 0\}$, acting on $\mathfrak{h}_0 \otimes \mathfrak{F}_{[0, \infty )}$.                         .

The process is \textbf{adapted} if, for each $t$, the operator $X(t)$ acts trivially on the future environment factor                .

QSDEs with adapted coefficients where originally introduced by Hudson \& Parthasarathy in 1984. Let $\{ X_{\alpha \beta } (t) : t \ge 0\}$ be four adapted quantum stochastic processes defined for $\alpha , \beta  \in \{ 0,1 \}$. We then define consider the QSDE
\begin{eqnarray}
\dot X (t) =
b(t)^\ast (t) X_{11} (t) b(t) +
b(t)^\ast X_{10} +
X_{01} (t) b(t)
+X_{00} (t),
\label{eq:fluxion_QSDE}
\end{eqnarray}
with initial condition $X(0) = X_0 \otimes I$. To understand this we take matrix elements between states of the form $| \phi \otimes \exp (\alpha ) \rangle$ and use the eigen-relation (\ref{eq:eigen_b}) to get the integrated form
\begin{eqnarray*}
\langle \phi \otimes \exp ( \alpha ) |  X (t) | \psi \otimes \exp (\beta ) \rangle = \langle \phi    | X_0  | \psi \rangle \, \langle \exp ( \alpha ) |\exp (\beta ) \rangle 
\end{eqnarray*}
\begin{eqnarray*}
+&& \int_0^t \alpha(\tau )^\ast\langle \phi \otimes \exp ( \alpha ) |  X_{11} (t) | \psi \otimes \exp (\beta ) \rangle  \beta (\tau ) d\tau \\
+&& \int_0^t \alpha(\tau )^\ast \langle \phi \otimes \exp ( \alpha ) |  X_{10} (t) | \psi \otimes \exp (\beta ) \rangle   d\tau \\
+&& \int_0^t \langle \phi \otimes \exp ( \alpha ) |  X_{01} (t) | \psi \otimes \exp (\beta ) \rangle   \beta (\tau ) d\tau \\
+&& \int_0^t \langle \phi \otimes \exp ( \alpha ) |  X_{00} (t) | \psi \otimes \exp (\beta ) \rangle     d\tau 
.
\end{eqnarray*}
Processes obtain this way are called quantum stochastic integrals.

The approach of Hudson and Parthasarathy is actually different. The arrive at the process defined by (\ref{eq:fluxion_QSDE}) by building the analogue of the Ito theory for stochastic integration: that is the show conditions in which
\begin{eqnarray}
dX(t) &=& X_{11}(t) \otimes d\Lambda (t)
+X_{10} (t) \otimes dB(t)^\ast
+X_{01} (t) \otimes dB (t)
+ X_{00} (t) \otimes dt ,\notag \\
\quad
\end{eqnarray}
makes sense as a limit process where all the increments are future pointing. That is $\Delta \Lambda \equiv \Lambda (t + \Delta t) - \Lambda (t)$ with $\Delta t >0$, etc.

One has, for instance,
\begin{eqnarray*}
&&\langle \phi \otimes \exp ( \alpha ) |  X_{00} (t) \otimes \Delta B (t)| \psi \otimes \exp (\beta ) \rangle    
\\ && \quad = \bigg( \int_t^{t+\Delta t}   \beta (\tau ) d\tau \bigg) \times  
\langle \phi \otimes \exp ( \alpha ) |  X_{00} (t) \otimes I| \psi \otimes \exp (\beta ) \rangle ,
\end{eqnarray*}
etc., so the two approaches coincide.

\subsection{Quantum Ito Rules}
It is clear from (\ref{eq:fluxion_QSDE}) that this calculus is Wick ordered - note that the creators $b(t)^\ast$ all appear to the left and all the annihilators, $b(t)$, appear to the right of the coefficients. The product of two Wick ordered expressions in not immediately Wick ordered and one must use the singular commutation relations to achieve this. This results in a additional term which corresponds to a quantum Ito correction. 

We have 
\begin{eqnarray*}
dB(t) dB(t) = dB(t)^\ast dB (t) = dB^\ast (t) dB^\ast (t)=0
\end{eqnarray*}
To see this, let $X_t $ adapted, then
\begin{eqnarray*}
\langle \exp (\alpha) | X_t dB(t)^\ast dB(t) | \exp (\beta) \rangle
= \alpha (t)^\ast \langle \exp ( \alpha ) | X_t \exp (\beta ) \rangle
\beta (t)
\,  (dt)^2
\end{eqnarray*}
As we have a square of $dt$ we can neglect such terms. 

However, we have
\begin{eqnarray*}
[ B(t) - B(s) , B (t)^\ast - B(s)^\ast ] = t-s,
\qquad (t>s)
\end{eqnarray*}
and so $\Delta B \, \Delta B^\ast = \Delta B^\ast
\Delta B + \Delta t$. The infinitesimal form of this is then
\begin{eqnarray*}
dB(t) dB(t)^\ast = dt .
\end{eqnarray*}
This is strikingly similar to the classical rule for increments of the Wiener process!

In fact, we have the following quantum Ito table
\begin{eqnarray*}
\begin{tabular}{l|llll}
$\times $ & $dt$ & $dB$ & $dB^{\ast }$ & $d\Lambda $ \\ \hline
$dt$ & 0 & 0 & 0 & 0 \\ 
$dB$ & 0 & 0 & $dt$ & $dB$ \\ 
$dB^{\ast }$ & 0 & 0 & 0 & 0 \\ 
$d\Lambda $ & 0 & 0 & $dB^{\ast }$ & $d\Lambda $%
\end{tabular}
\end{eqnarray*}
Each of the non-zero terms arises from multiplying two processes that are not in Wick order.

For a pair of quantum stochastic integrals, we have the following quantum Ito product formula
\begin{eqnarray*}
d \big( XY \big)
=(dX) dY + dX (dY) + (dX) (dY).
\end{eqnarray*}
Unlike the classical version, the order of $X$ and $Y$ here is crucial.

\subsection{Some \guillemotleft Classical Processes\guillemotright \, On Fock Space}
The process $Q(t) = B(t) + B(t)^\ast $ is self-commuting, that is $[Q(t),Q(s) ] =0, \quad \forall t,s$, and has the distribution of a Wiener process is the vacuum state
\begin{eqnarray*}
\langle \dot Q (t) \rangle &=& \langle \Omega
| [
b(t) + b(t)^\ast ] \Omega \rangle =0, \\
\langle \dot Q(t) \dot Q(s) \rangle
&=& \langle \Omega | b(t) b^\ast (s) \Omega \rangle
= \delta (t-s) .
\end{eqnarray*}

The same applies to $P(t) = \frac{1}{i} [ B(t) - B(t)^\ast ]$, but
\begin{eqnarray*}
[Q(t), P(s)] =2i \, \text{min}(t,s).
\end{eqnarray*}
So we have two non-commuting Wiener processes in Fock space. We refer to $Q$ and $P$ as canonically conjugate quadrature processes.

One see that, for instance,
\begin{eqnarray*}
dQ dQ = dB dB^\ast = dt.
\end{eqnarray*}

We also obtain a Poisson process by the prescription
\begin{eqnarray*}
N(t) = \Lambda (t) + \sqrt{\nu} B^\ast (t)
+\sqrt{\nu} B(t) + \nu t.
\end{eqnarray*}
One readily checks that $dN dN = dN$ from the quantum Ito table.

\subsection{Emission-Absorption Interactions}
Let us consider a singular Hamiltonian of the form
\begin{eqnarray}
\Upsilon (t) =
H \otimes I
+i
L \otimes b(t)^\ast - i L^\ast \otimes b(t).
\label{eq:Upsilon}
\end{eqnarray}
We will try and realize the solution to the Schr\"{o}dinger equation
\begin{eqnarray}
\dot U (t) = -i \Upsilon (t) \, U(t), \qquad U(0)=I.
\label{eq:QSDE_flux_Up}
\end{eqnarray}
as a unitary quantum stochastic integral process.

Let us first remark that the annihilator part of (\ref{eq:Upsilon}) will appear out of Wick order when we consider (\ref{eq:QSDE_flux_Up}). The standard approach in quantum field theory is to develop the unitary $U(t)$ as a Dyson series expansion - often re-interpreted as a time order-exponential:
\begin{eqnarray*}
U(t) &=& I -i \int_0^t \Upsilon (\tau ) U(\tau ) 
d \tau \\
&=&  1 - i \int_0^t d\tau \Upsilon (\tau )
+ (-i)^2 \int_0^t d \tau_2 \int_0^{\tau_2} d\tau_2
\Upsilon (\tau _2 ) \Upsilon (\tau _1 ) + \cdots\\
&=& \vec{T} e^{-i \int_0^t \Upsilon ( \tau ) d\tau} .
\end{eqnarray*}
In our case the field terms - the quantum white noises - are linear, however, we have the problem that they come multiplied by the system operators $L$ and $L^\ast$ which do not commute, and don't necessarily commute with $H$ either.

Fortunately we can do the Wick ordering in one fell swoop rather than having to go down each term of the Dyson series.
We have
\begin{eqnarray*}
\left[ b\left( t\right) ,U\left( t\right) \right]  &=&\left[ b\left(
t\right) ,I-i\int_{0}^{t}\Upsilon \left( \tau \right) U\left( \tau \right)
d\tau \right] =-i\int_{0}^{t}\left[ b\left( t\right) ,\Upsilon \left( \tau
\right) \right] U\left( \tau \right) d\tau  \\
&=&\int_{0}^{t}\left[ b\left( t\right) ,Lb\left( \tau \right) ^{\ast }\right]
U\left( \tau \right) d\tau \\
&=&L\int_{0}^{t}\delta \left( t-\tau \right)
U\left( \tau \right) d\tau =\frac{1}{2}LU\left( t\right) ,
\end{eqnarray*}
where we dropped the $[b(t) , U(\tau )]$ term as this should vanish for $t> \tau$ and took half the weight of the $\delta$-function due to the upper limit $t$ of the integration. However, we get
\begin{eqnarray*}
b\left( t\right) U\left( t\right) =U\left( t\right) b\left( t\right) +\frac{1%
}{2}LU\left( t\right) .
\end{eqnarray*}
Plugging this into the equation (\ref{eq:QSDE_flux_Up}), we get
\begin{eqnarray*}
\dot{U}\left( t\right)  &=&b\left( t\right) ^{\ast }LU\left( t\right)
-L^{\ast }b\left( t\right) U\left( t\right) -iH\left( t\right) U\left(
t\right)  \\
&=&b\left( t\right) ^{\ast }LU\left( t\right) -L^{\ast }U\left( t\right)
b\left( t\right) -\left( \frac{1}{2}L^{\ast }L+iH\right) U\left( t\right) .
\end{eqnarray*}
which is now Wick ordered. We can interpret this as the Hudson-Parthasarathy equation
\begin{eqnarray*}
dU\left( t\right) =\left\{ L\otimes dB\left( t\right) ^{\ast }-L^{\ast
}\otimes dB\left( t\right) -\left( \frac{1}{2}L^{\ast }L+iH\right) \otimes
dt\right\} U\left( t\right) .
\end{eqnarray*}

The corresponding Heisenberg equation for $j_t (X) = U(t)^\ast [X \otimes I ] U(t) $ will be

\begin{eqnarray*}
dj_{t}\left( X\right)  &=&dU\left( t\right) ^{\ast }\left[ X\otimes I\right]
U\left( t\right) +U\left( t\right) ^{\ast }\left[ X\otimes I\right] dU\left(
t\right) \\
&& +dU\left( t\right) ^{\ast }\left[ X\otimes I\right] dU\left(
t\right)  \\
&=&j_{t}\left( \mathcal{L}X\right) \otimes dt+j_{t}\left( \left[ X,L\right]
\right) \otimes dB\left( t\right) ^{\ast }+j_{t}\left( \left[ L^{\ast },X%
\right] \right) \otimes dB\left( t\right) 
\end{eqnarray*}
where
\begin{eqnarray*}
\mathcal{L}X &=&-X\left( \frac{1}{2}L^{\ast }L+iH\right) -\left( \frac{1}{2}%
L^{\ast }L-iH\right) X+L^{\ast }XL \\
&=&\frac{1}{2}\left[ L^{\ast },X\right] L+\frac{1}{2}L^{\ast }\left[ X,L%
\right] -i\left[ X,H\right] .
\end{eqnarray*}
We note that we obtain the typical Lindblad form for the generator.

\subsection{Scattering Interactions}
We mention that we could also treat a Hamiltonian with only scattering terms
Let us set $\Upsilon \left( t\right) =E\otimes b\left( t\right) ^{\ast }b\left(
t\right) $. The same sort of argument leads to
\begin{eqnarray*}
\left[ b\left( t\right) ,U\left( t\right) \right] =-iE\int_{0}^{t}\left[
b\left( t\right) ,b\left( \tau \right) ^{\ast }\right] b\left( \tau \right)
U\left( \tau \right) d\tau =-\frac{i}{2}Eb\left( t\right) U\left( t\right) ,
\end{eqnarray*}
which can be rearranged to give
\begin{eqnarray*}
b\left( t\right) U\left( t\right) =\frac{1}{I-\frac{i}{2}E}U\left( t\right)
b\left( t\right) .
\end{eqnarray*}
So the Wick ordered form is
\begin{eqnarray*}
\dot{U}\left( t\right) =Eb\left( t\right) ^{\ast }b\left( t\right) U\left(
t\right) =\frac{E}{I-\frac{i}{2}}b\left( t\right) ^{\ast }U\left( t\right)
b\left( t\right) 
\end{eqnarray*}
or in quantum Ito form
\begin{eqnarray*}
dU\left( t\right) =\left( S-I\right) \otimes d\Lambda \left( t\right)
\,U\left( t\right) ,\qquad \left( S=\frac{I+\frac{i}{2}E}{I-\frac{i}{2}E}%
\text{, unitary!}\right) .
\end{eqnarray*}
The Heisenberg equation here is $dj_{t}\left( X\right) =j_{t}\left( S^{\ast
}XS-X\right) \otimes d\Lambda \left( t\right) $.

This is all comparable to the classical Poisson process driven evolution involving unitary kicks.

\subsection{The \guillemotleft SLH Formalism\guillemotright }
The examples considered up to now used only one species of quanta. We could in fact have $n$ channels, based on $n$ quantum white noises:
\begin{eqnarray*}
 [  b_j (t) ,
b^\ast_k (s) ] = \delta_{jk} \, \delta (t-s) .
\end{eqnarray*}

The most general form of a unitary process with fixed coefficients may be described as follows:
we have a \textbf{Hamiltonian} $H=H^\ast$, a column vector of coupling/ collapse operators
\begin{eqnarray*}
 L=\left[
\begin{array}{c}
 L_{1} \\
\vdots  \\
L_{n}
\end{array}
\right] ,
\end{eqnarray*}
and a matrix of operators
\begin{eqnarray*}
 S=\left[
\begin{array}{ccc}
S_{11}  &\cdots  &  S_{1n}  \\
\vdots  & \ddots  & \vdots  \\
 S_{n1}  & \cdots  &  S_{nn} 
\end{array}
\right] ,
\qquad   S^{-1} =  S^\ast .
\end{eqnarray*}

For each such triple $(S,L,H)$ we have the QSDE
\begin{eqnarray}
 d U(t) &=& \bigg\{
\sum_{jk}(  S_{jk} - \delta_{jk} I)\otimes
d \Lambda_{jk} (t)
+ \sum_j L_j  \otimes dB_j^\ast (t) \\
&&
- \sum_{jk} L_j^\ast S_{jk} \otimes dB_k (t) 
 -(\frac{1}{2} \sum_k L_k^\ast L_k +
i H ) \otimes dt \bigg\} \, U(t)
\label{eq:SLH_QSDE}
\end{eqnarray}
which has, for initial condition $U(0) =I$, a solution which is a unitary adapted quantum stochastic process. The emission-absorption case is the $n=1$ model with no scattering ($S=I$). Likewise the purse scattering corresponds to $H=0$ and $L=0$.

System observables evolve according to the Heisenberg-Langevin equation
\begin{eqnarray*}
d j_t ( X) &=& \sum_{jk}
j_t 
(S^\ast_{lj}XS_{lk}- \delta_{jk} X)
d \Lambda_{jk} (t)
+ \sum_{jl}j_t ( S_{lj}^\ast [L_l,X])\otimes dB_j (t)^\ast \\
&&  + \sum_{lk} j_t ([X,L^\ast_l ] S_{lk}) \otimes dB_k (t)
+j_t ( \mathscr{L} X)\otimes dt .
\end{eqnarray*}
where the generator is the traditional Lindblad form
\begin{eqnarray*}
\mathscr{L} X = \frac{1}{2}\sum_k L^\ast_k [X,L_k]
+ \frac{1}{2} \sum_k [L^\ast_k  , X] L_k -i [X,H ] .
\end{eqnarray*}

\subsection{Quantum Outputs}
The output fields are defined by
\begin{eqnarray*}
B^{\text{out}}_k (t) =
U(t)^\ast [ I \otimes B_k (t) ] U(t).
\end{eqnarray*}

From the quantum Ito calculus we find that
\begin{eqnarray*}
dB^{\text{out}}_j (t) = \sum_k j_t (S_{jk}) \otimes dB_k (t)
+ j_t (L_k) \otimes dt ,
\end{eqnarray*}
Or, maybe more suggestively in quantum white noise language,
\begin{eqnarray*}
b^{\text{out}}_j (t) = \sum_j j_t (S_{jk})
 \otimes b_k(t) +j_t (L_j) \otimes I.
\end{eqnarray*}

\section{Quantum Filtering}
We now set up the quantum filtering problem. For simplicity, we will take $n=1$ and set $S=I$ so that we have a simple emission-absorption interaction. We will also consider the situation where we measure the $Q$-quadrature of the output.

The initial state is taken to be $|\psi_0 \rangle \otimes |\Omega \rangle $, and in the Heisenberg picture this is fixed for all time.

The analogue of the stochastic dynamical equation considered in the classical filtering problem is the Heisenberg-Langevin equation
\begin{eqnarray*}
dj_{t}\left( X\right)  =j_{t}\left( \mathcal{L}X\right) \otimes dt+j_{t}\left( \left[ X,L\right]
\right) \otimes dB\left( t\right) ^{\ast }+j_{t}\left( \left[ L^{\ast },X%
\right] \right) \otimes dB\left( t\right) 
\end{eqnarray*}
where $\mathcal{L}X =\frac{1}{2}\left[ L^{\ast },X\right] L+\frac{1}{2}L^{\ast }\left[ X,L%
\right] -i\left[ X,H\right] $.

Some care is needed in specifying what exactly we measure: we should really work in the Heisenberg picture for clarity. The $Q$-quadrature of the input field is 
$Q\left( t\right) =B\left( t\right) +B\left( t\right) ^{\ast }$ which we have already seen is a Wiener process for the vacuum state of the field. Of course this is not what we measure - we measure the output quadrature!

Set
\begin{eqnarray*}
Y^{\text{in}}\left( t\right) =I\otimes Q\left( t\right) .
\end{eqnarray*}
As indicated in our discussion on von Neumann's measurement model, what we actually measure is
\begin{eqnarray*}
Y^{ \text{out}} (t) = U(t)^\ast Y^{\text{in}} (t) U(t) = B^{\text{out} } (t) + B^{\text{out}} (t)^\ast .
\end{eqnarray*}
The differential form of this is
\begin{eqnarray*}
dY^{\text{out}} (t) = dY^{\text{in}} (t) + j_t (L+L^\ast ) dt .
\end{eqnarray*}
Note that
\begin{eqnarray*}
dY^{\text{in}}\left( t\right) dY^{\text{in}}\left( t\right) =dt=dY^{\text{out%
}}\left( t\right) dY^{\text{out}}\left( t\right) .
\end{eqnarray*}

The dynamical noise is generally a quantum noise and can only be considered classical in very special circumstances, while the observational noise is just its $Q$-quadrature which can hardly be treated as independent!

In complete contrast to the classical filtering problem we considered earlier, we have no paths for the system - just evolving observables of the system. What is more these observables do not typically commute amongst themselves, or indeed the measured process.

We can only apply Bayes Theorem in the situation where the quantities involved have a joint probability distribution, and in the quantum world this requires them to be compatible. At this stage it may seem like a miracle that we have any theory of filtering in the quantum world. However, let us stake stock of what we have.

\subsection{What Commutes With What?}

For fixed $s \ge 0$, let $U(t,s)$ be the solution to the QSDE (\ref{eq:SLH_QSDE}) in time variable $ t \ge s$ with $U(s,s)=I$. Formally, we have 
\begin{eqnarray*}
U\left( t,s\right) =\vec{T}e^{-i\int_{s}^{t}\Upsilon \left( \tau \right) d\tau }
\end{eqnarray*}
which is the unitary which couples the system to the part of the field that enters over the time $s\leq
\tau \leq t$. In terms of our previous definition, we have $U(t) = U(t,0)$ and we have the property
\begin{eqnarray*}
U\left( t\right) =U\left( t,s\right) U\left( s\right) ,\qquad \left(
t>s>0\right) .
\end{eqnarray*}

In the Heisenberg picture, the observables evolve
\begin{eqnarray*}
j_{t}\left( X\right)  &=&U\left( t\right) ^{\ast }\left[ X\otimes I\right]
U\left( t\right) , \\
Y^{\text{out}}\left( t\right)  &=&U\left( t\right) ^{\ast }\left[ I\otimes
Q\left( t\right) \right] U\left( t\right) .
\end{eqnarray*}

We know that the input quadrature is self-commuting, but what about the output one?
A key identity here is that 
\begin{eqnarray*}
Y^{\text{out}}\left( t\right) =U\left( t\right) ^{\ast }Y^{\text{in}}\left(
s\right) U\left( t\right) ,\qquad \left( t>s\right) ,
\end{eqnarray*}
which follows from the fact that $\left[ Y^{\text{in}}\left( s\right)
,U\left( t,s\right) \right] =0$. 

\bigskip 

From this, we see that the process $Y^{\text{out}}$ is also commutative since
\begin{eqnarray*}
\left[ Y^{\text{out}}\left( t\right) ,Y^{\text{out}}\left( s\right) \right]
=U\left( t\right) ^{\ast }\left[ Y^{\text{in}}\left( t\right) ,Y^{\text{in}%
}\left( s\right) \right] U\left( t\right) =0,\quad \left( t>s\right) .
\end{eqnarray*}
If this was not the case then subsequent measurements of the process $Y^{\text{out}}$ would invalidate (disturb?) earlier ones. In fancier parlance, we say that process is \textbf{not self-demolishing} - that is, all parts are compatible with each other.

A similar line of argument shows that 
\begin{eqnarray*}
\left[ j_{t}\left( X\right) ,Y^{\text{out}}\left( s\right) \right] =U\left(
t\right) ^{\ast }\left[ X\otimes I,I\otimes Q\left( t\right) \right] U\left(
t\right) =0,\quad \left( t>s\right) .
\end{eqnarray*}
Therefore, we have a joint probability for $j_{t}\left( X\right) $ and the continuous collection of observables $\left\{
Y^{\text{out}}\left( \tau \right) :0\leq \tau \leq t\right\} $ so can use
Bayes Theorem to estimate $j_t (X)$ for any $X$ using the past observations.
Following V.P. Belavkin, we refer to this as the \textbf{non-demolition principle}.

\subsection{The Conditioned State}
In the Schr\"{o}dinger picture, the state at time $t \ge 0$ is $|\Psi _{t}\rangle =U\left( t\right) |\phi \otimes
\Omega \rangle $, so
\begin{eqnarray*}
d|\Psi _{t}\rangle  &=&-\left( \frac{1}{2}L^{\ast }L+iH\right) |\Psi
_{t}\rangle dt+LdB\left( t\right) ^{\ast }|\Psi _{t}\rangle -L^{\ast
}dB\left( t\right) |\Psi _{t}\rangle  \\
&=&-\left( \frac{1}{2}L^{\ast }L+iH\right) |\Psi _{t}\rangle dt+LdB\left(
t\right) ^{\ast }|\Psi _{t}\rangle +LdB\left( t\right) |\Psi _{t}\rangle  \\
&=&-\left( \frac{1}{2}L^{\ast }L+iH\right) |\Psi _{t}\rangle dt+LdY^{\text{in}} (t)|\Psi
_{t}\rangle .
\end{eqnarray*}

Here we have used a profound trick due to A.S. Holevo. The differential $dB(t)$ acting on $ | \Psi_t \rangle$ yields zero since it is future pointing and so only affects the future part which, by adaptedness, is the vacuum state of the future part of the field. To get from the first line to the second line, we remove and add a term that is technically zero. In its reconstituted form, we obtain the $Q$-quadrature of the input. The result is that we obtain an expression for the state $| \Psi_t \rangle$ which is \lq\lq diagonal\rq\rq \, in the input quadrature - our terminology here is poor (we are talking about a state not and observable!) but hopefully wakes up physicists to see what's going on.

The above equation is equivalent to the SDE in the system Hilbert space
\begin{eqnarray}
d|\chi _{t}\rangle =-\left( \frac{1}{2}L^{\ast }L+iH\right) |\chi
_{t}\rangle dt+L|\chi _{t}\rangle dy_{t}
\label{eq:BZ}
\end{eqnarray}
where $\mathbf{y}$ is a sample path - or better still, \guillemotleft eigen-path\guillemotright \, - of the quantum stochastic process $Y^{\text{in}}$.

We refer to (\ref{eq:BZ}) as the \textbf{Belavkin-Zakai equation}.

\subsection{The Quantum Filter}
Let us begin with a useful computational
\begin{eqnarray}
\langle \phi \otimes \Omega |j_{t}\left( X\right) F\left[ Y_{\left[ 0,t%
\right] }^{\text{out}}\right] |\phi \otimes \Omega \rangle 
&=&
\langle \phi \otimes \Omega | U(t)^\ast \big(  X\otimes  F\left[ Y_{\left[ 0,t%
\right] }^{\text{in}}\right] \big) U(t) |\phi \otimes \Omega \rangle \notag \\
&=&
\langle \Psi_t |   X\otimes  F\left[ Y_{\left[ 0,t%
\right] }^{\text{in}}\right]  |\Psi_t \rangle \notag \\
 &=&
\int
\langle \chi_t (\mathbf{y} ) |  X\otimes |\chi_t (\mathbf{y}) \rangle \, F\left[ \mathbf{y}\right]  
\, \mathbb{P}_{\text{Wiener}} [d \mathbf{y}]. \notag  \\
&& \quad
\label{eq:big}
\end{eqnarray}

A few comments are in order here. The operator $j_{t}\left( X\right) $ will commute with any functional of the past measurements - here $F\left[ Y_{\left[ 0,t\right] }^{\text{out}}\right] $. In the first equality is pulling things back in terms of the unitary $U(t)$. The second is just the equivalence between Schr\"{o}dinger and Heisenberg pictures. The final one just uses the equivalent form (\ref{eq:BZ}): note that the paths of the input quadrature gets their correct weighting as Wiener processes.

Setting $X=I$ in (\ref{eq:big}), we get the
\begin{eqnarray*}
\langle \phi \otimes \Omega | F\left[ Y_{\left[ 0,t%
\right] }^{\text{out}}\right] |\phi \otimes \Omega \rangle 
&=&
\int
\langle \chi_t (\mathbf{y} |  \chi_t (\mathbf{y} )\rangle \, F\left[ \mathbf{y}\right]  
\, \mathbb{P}_{\text{Wiener}} [d \mathbf{y}]
\end{eqnarray*}
So the probability of the measured paths is
\begin{eqnarray*}
\mathbb{Q} [d \mathbf{y}] =
\langle \chi_t (\mathbf{y} )|  \chi_t (\mathbf{y}) \rangle \, \mathbb{P}_{\text{Wiener}} [d \mathbf{y}] .
\end{eqnarray*}
Now this last equation deserves some comment! The vector $| \Psi_t \rangle$, which lives in the system tensor Fock space, is properly normalized, but its corresponding form $ | \chi _t \rangle$ is not! The latter is a stochastic process taking values in the system Hilbert space and is adapted to input quadrature. However, we never said that $| \chi_t \rangle$ had to be normalized too, and indeed it follows from or \lq\lq diagonalization\rq\rq \, procedure. In fact, if $| \chi_t \rangle$ was normalized then the output measure would follow a Wiener distribution and  so we would be measuring white noise!

From (\ref{eq:big}) again, we an deduce the filter: we get (using the arbitrariness of the functional $F$)
\begin{eqnarray}
\pi_t (X) = \frac{ \langle \chi_t (\mathbf{y}) | X | \chi_t (\mathbf{y}) \rangle}{
\langle \chi_t (\mathbf{y} )|  \chi_t (\mathbf{y}) \rangle } .
\end{eqnarray}

This has a remarkable similarity to (\ref{eq:filter_pi}). Moreover, using the Ito calculus see that
\begin{eqnarray*}
d \langle \chi_t (\mathbf{y}) | X | \chi_t (\mathbf{y}) \rangle
&=&
\langle \chi_t (\mathbf{y}) | \mathcal{L} X | \chi_t (\mathbf{y}) \rangle  dt \notag \\
&&+\langle \chi_t (\mathbf{y}) |  \big( XL+L^\ast X \big) |\chi_t (\mathbf{y} )\rangle \, dy(t).
\end{eqnarray*}
This is the quantum analogue of the Duncan-Mortensen-Zakai equation.

So small work is left in order to derive the filter equation. We first observe that the normalization (set $X=I$) is that
\begin{eqnarray*}
d \langle \chi_t (\mathbf{y}) |  \chi_t (\mathbf{y}) \rangle
=\langle \chi_t (\mathbf{y}) |  \big( L + L^\ast  \big) |\chi_t (\mathbf{y} )\rangle \, dy(t).
\end{eqnarray*}
Using the Ito calculus, it is then routine to show that the quantum filter is
\begin{eqnarray}
 d \pi_t (X) =  \pi_t ( \mathcal{L}X) \, dt + \big\{ \pi_t (XL+L^\ast X ) - \pi_t (X) \pi_t (L+L^\ast ) \big\} dI(t)
\end{eqnarray}
where the innovations are defined by
\begin{eqnarray}
 dI (t) = dY^{\text{out}} (t) - \pi_t (L+L^\ast ) \, dt .
\end{eqnarray}
Again, the innovations have the statistics of a Wiener process. As in the classical case, the innovations give the difference between what we observe next, $dY^{\text{out}} (t)$, and what we would have expected based on our observations up to that point, $\pi_t (L+L^\ast ) \, dt $. The fact that the innovations are a Wiener process is a reflection of the efficiency of the filter - after extracting as much information as we can out of the observations, we are left with just white noise.

\textbf{Acknowledgements}
I would like to thank the staff at CIRM, Luminy, and at Institut Henri Poincar\'{e} for their kind support for this workshop. I am also grateful to the other organizers Pierre Rouchon and Denis Bernard for valuable comments during the writing of these notes.

\end{document}